\documentclass{amsart}[12pt]
\usepackage{amsmath}
\usepackage{amssymb}
\usepackage{amsfonts}
\usepackage{graphicx}
\usepackage{texdraw}
\usepackage{tikz}
\usepackage{enumitem}
\usepackage{graphpap}
\usepackage{wasysym}
\usepackage{color}
\usepackage{pxfonts}
\usepackage[OT2,T1]{fontenc}

\newtheorem{thm}{Theorem}[section]

\newtheorem{lem}[thm]{Lemma}

\theoremstyle{definition}

\theoremstyle{remark}
\newtheorem{rem}{Remark}

\numberwithin{equation}{section}

\newcommand{\wfun}{{\omega}}

\newcommand{\Prob}{{\mathbb{P}}}
\renewcommand{\d}{\partial}
\newcommand{\dbar}{\bar{\partial}}
\newcommand{\calA}{{\mathcal A}}
\newcommand{\calB}{{\mathcal B}}

\newcommand{\calO}{{\mathcal O}}
\newcommand{\calG}{{\mathcal G}}
\newcommand{\calH}{{\mathcal H}}

\newcommand{\calN}{{\mathcal N}}

\newcommand{\calE}{{\mathcal E}}

\newcommand{\calW}{{\mathcal W}}

\newcommand{\bigO}{{\mathcal O}}

\newcommand{\T}{{\mathbb T}}

\newcommand{\1}{{\mathbf{1}}}

\newcommand{\trace}{\operatorname{trace}}
\newcommand{\Subh}{\operatorname{Subh}}
\newcommand{\fluct}{\operatorname{fluct}}
\renewcommand{\Cap}{\operatorname{Cap}}
\newcommand{\R}{{\mathbb R}}
\newcommand{\Z}{{\mathbb Z}}
\newcommand{\E}{{\mathbb E}}

\newcommand{\C}{{\mathbb C}}
\newcommand{\D}{{\mathbb D}}

\newcommand{\para}{\lambda}

\newcommand{\supp}{\operatorname{supp}}
\newcommand{\re}{\operatorname{Re}}
\newcommand{\im}{\operatorname{Im}}

\newcommand{\erfc}{\operatorname{erfc}}
\newcommand{\Int}{\operatorname{Int}}
\newcommand{\Ext}{\operatorname{Ext}}

\newcommand{\dist}{\operatorname{dist}}
\newcommand{\const}{\operatorname{const.}}

\newcommand{\He}{\operatorname{He}}
\newcommand{\eps}{{\varepsilon}}

\makeatletter
\newcommand*\bigcdot{\mathpalette\bigcdot@{.5}}
\newcommand*\bigcdot@[2]{\mathbin{\vcenter{\hbox{\scalebox{#2}{$\m@th#1\bullet$}}}}}
\makeatother

\allowdisplaybreaks

\begin{document}



	\title[Fluctuations of Coulomb systems]{On fluctuations of Coulomb systems and universality of the Heine distribution}

	\begin{abstract} We consider a class of external potentials on the complex plane $\C$ for which the
		coincidence set to the obstacle problem contains a Jordan curve in the exterior of the droplet. We refer to this curve as a spectral outpost. We study the corresponding Coulomb gas at $\beta=2$.

		Generalizing recent work in the radially symmetric case, we prove that the number of particles which fall near the spectral outpost has an asymptotic Heine distribution, as the number of particles $n\to\infty$.
		
		We also consider a class of potentials with disconnected droplets whose connected components are separated by a ring-shaped spectral gap.
		We prove that the fluctuations of the number of particles that fall near a given component has an asymptotic discrete normal distribution, which depends on $n$.

		For the case of disconnected droplets we also consider fluctuations of general smooth linear statistics and show that they tend to distribute as the sum of a Gaussian field and an independent, oscillatory, discrete Gaussian field.

		Our methods involve a new asymptotic formula on the norm of monic orthogonal polynomials in the bifurcation regime
		and a variant of the method of limit Ward identities of Ameur, Hedenmalm, and Makarov.

	\end{abstract}

	\author{Yacin Ameur}
	\address{Yacin Ameur\\
		Department of Mathematics\\
		Lund University\\
		22100 Lund, Sweden}
	\email{ Yacin.Ameur@math.lu.se}

	\author{Joakim Cronvall}
	\address{Joakim Cronvall\\
		Department of Mathematics\\
		Lund University\\
		22100 Lund, Sweden}
	\email{Joakim.Cronvall@math.lu.se}

	\maketitle

	\section{Introduction}\label{intro} In the theory of Hermitian random matrices it is natural to study
	the number of eigenvalues which fall in the vicinity of a given connected component of the droplet (i.e.~ the support of the equilibrium measure).
	The connected components are compact intervals called ``cuts''; accordingly, disconnected droplets are known as the ``multi-cut regime''. It has been well-studied, see e.g. \cite{DIZ,CFWW,ClaeysGravaMcLaughlin,BG2,Fo} and references therein.

	In dimension two (eigenvalues of normal random matrices) the regime of disconnected droplets is no less natural, but is a relatively new area of research.
	
	A natural starting point
	is to study radially symmetric potentials; this was first done in a hard edge setting in \cite{C}.
	Another natural model, with discrete rotational symmetry, also appears in the literature, following \cite{BM}. Recently, some other examples of disconnected droplets also appeared in connection with insertion of point charges, e.g. see \cite{BY2}.

	In the radially symmetric case a disconnected droplet is a union of concentric annuli, one of which is perhaps a disc. In the soft edge regime,  a comprehensive study of such ensembles is carried out in our earlier works \cite{ACC1,ACC2}.

	In the present work we study for the first time general smooth linear statistics related to a class of disconnected droplets (with soft edges) with ring-shaped gaps obeying only some natural potential theoretic compatibility conditions.
	
	We shall identify two universality classes that we term
	``spectral outpost'' and ``spectral gap''.
	(Cf.~Fig.~\ref{fig1} and Fig.~\ref{fig2} below.)

	For the first type (spectral outpost) the droplet is connected, but the coincidence set for the obstacle problem contains a Jordan curve sitting in the exterior of the droplet. We refer to this curve as an outpost.
	
	The second type (spectral gap ensemble) corresponds to disconnected droplets whose connected components are
	separated by a ring-shaped spectral gap.
	
	For these classes, the uncertainty in the number of eigenvalues near a component of the droplet (or the outpost) are given by certain
	Heine distributions, which we interpret as the number of eigenvalues that are displaced from one component, across a spectral gap, to another. These Heine distributions carry certain universal potential theoretic and geometric information. For a ring-shaped spectral gap separating two components of the droplet, we obtain two Heine distributions which ``oscillate'' with the number of eigenvalues, and their difference has an oscillatory discrete Gaussian distribution.

	Before turning to the details, it is convenient to fix notation and recall some background on Coulomb gas ensembles in dimension two. We refer to \cite{BF} as a general source on Coulomb gas ensembles in dimension two.

	\subsection{Coulomb gas ensembles} \label{cge}

	By a confining potential, we mean a lower semicontinuous function $Q:\C\to\R\cup\{+\infty\}$ which is finite on some set of positive capacity and satisfies the growth condition
	\begin{equation}\label{growth_Q}
	\liminf_{z\to\infty}\frac {Q(z)}{\log|z|^2}>1.
	\end{equation}

	Given a large integer $n$, the Coulomb gas (at $\beta=2$) with respect to $Q$ is a random sample $\{z_j\}_1^n$ from the Gibbs measure on $\C^n$
	\begin{align}\label{gibbs}d\Prob_n(z_1,\ldots,z_n)=\frac 1 {Z_n}\int_{\C^n}\prod_{1\le i< j\le n}|z_i-z_j|^2\prod_{i=1}^n e^{-nQ(z_i)}\, dA(z_i)\end{align}
	where $dA=dxdy/\pi$ denotes the normalized area measure on $\C$ and $Z_n=Z_n(Q)$ is the partition function, i.e., the normalizing constant making $\Prob_n$ a probability measure.

	As is well known, the sample can be identified with eigenvalues of normal random matrices with respect to a certain weighted distribution.

	We now recall some basic potential theoretic objects, such as ``droplet'' and ``coincidence set''.

	Given a compactly supported Borel probability measure $\mu$ on $\C$, the $Q$-weighted logarithmic energy is defined by
	$$I_Q[\mu]=\int_{\C^2}\log \frac 1 {|z-w|}d\mu(z)\, d\mu(w)+\mu(Q),$$
	where we write $\mu(f):=\int f\, d\mu$. As is well-known \cite{ST} there exists a unique compactly supported Borel probability measure $\sigma=\sigma_Q$ which minimizes
	$I_Q[\mu]$ over all such $\mu$.
	
	This measure $\sigma$ is known as Frostman's equilibrium measure; it provides the classical approximation of the Coulomb gas; for instance, the empirical measures $n^{-1}\sum_1^n\delta_{z_j}$ converge in a certain probabilistic sense to $\sigma$ as $n\to\infty$. (See e.g.~\cite{D}.)

	The support $$S=S[Q]:=\supp\sigma$$
	is called the droplet in potential $Q$.
	
	Assuming that $Q$ is $C^2$-smooth in a neighbourhood of $S$, the equilibrium measure is absolutely continuous and has the apriori structure \cite{ST}
	$$d\sigma(z)=\Delta Q(z)\cdot \1_S(z)\, dA(z)$$
	where we write $\Delta=:\d\dbar$ for the normalized Laplacian.
The complex derivatives $\d$ and $\dbar$ are given  by
	$\d=\frac 1 2(\d_x+\frac 1 i\d_y)$ and $\dbar=\frac 1 2(\d_x-\frac 1 i\d_y)$ where $z=x+iy$.

	We next recall a few facts pertaining to the obstacle problem associated with the obstacle $Q(z)$.

	For this, we introduce the obstacle function $\check{Q}(z)$ by
	$$\check{Q}(z)=\sup\{f(z)\,;\, f\in \operatorname{Subh}_{Q,1}\},$$
	where $\Subh_{Q,1}$ denotes the class of subharmonic functions $f:\C\to\R$ such that $f\le Q$ everywhere and $f(z)=2\log|z|+\bigO(1)$ as $z\to\infty$.
	
	It is well-known that $\check{Q}$
	belongs to the class $C^{1,1}(\C)$ of differentiable functions with Lipschitz continuous gradients and that $\check{Q}$ is harmonic in $\C\setminus S$, cf.~\cite{ST}. Also $\check{Q}$ belongs to $\Subh_{Q,1}$.
	
	The coincidence set $S^*$ is defined by
	$$S^*=S^*[Q]:=\{z\in\C\,;\, Q(z)=\check{Q}(z)\}.$$

	We always have the inclusion $S\subset S^*$, cf. \cite{ST}.
	In the sequel, we assume that $Q$ is smooth in a neighbourhood of $S^*$ and
	obeying the strict subharmonicity $\Delta Q>0$ along $\d S^*$. Following \cite{ACC2}, we refer to a connected component of $S^*\setminus S$ as a \textit{spectral outpost}.

	Under these conditions, an outpost has area zero, but might very well have positive capacity. It turns out that outposts of positive capacity affect the fluctuations of the Coulomb gas in an essential way, and by studying this effect, we will find tools to handle droplets with ring-shaped spectral gaps.

	\medskip
	
	\textit{In what follows, the external potential $Q(z)$ is assumed to satisfy the conditions above.}

	\subsubsection*{General notation} $\hat{\C}=\C\cup\{\infty\}$ is the Riemann sphere,
	$\D_e=\{|w|>1\}\cup\{\infty\}$ is the exterior disc, $\T=\{z\in\C\,;\,|z|=1\}$ is the unit circle.

	If $\Gamma$ is a Jordan curve, we define $\Int \Gamma$ and $\Ext \Gamma$ as the bounded and unbounded component of $\hat{\C}\setminus \Gamma$ respectively.

	If $E$ is a bounded subset of $\C$ we write $\d_* E$ for its outer boundary, i.e., the boundary of the unbounded component of $\C\setminus E$.

	By ``smooth'', we always mean $C^\infty$-smooth.

	\subsection{Linear statistics} Let $\{z_j\}_1^n$ be a random sample from  \eqref{gibbs}. For suitable (say continuous and bounded) real-valued functions $f$ we introduce the random variables (or linear statistics)
	$$\trace_n f=\sum_{j=1}^n f(z_j),\qquad \fluct_n f=\sum_{j=1}^n f(z_j)-n\sigma(f).$$

	A basic result in \cite{AM} asserts that if the droplet $S$ is connected, $S=S^*$, and the boundary $\d S$ is everywhere smooth, then the fluctuations have an asymptotic Gaussian distribution. Continuing from previous work in \cite{ACC1,ACC2} we shall here study fluctuations in some cases where $S^*$ is disconnected.

	Our main tool
	is the cumulant generating function of $\fluct_n f$
	\begin{equation}\label{cgf}F_{n,f}(s):=\log \E_n \exp (s\fluct_n f),\qquad (s\in\R).\end{equation}

	The function \eqref{cgf} is closely related to the partition function with respect to perturbed potentials
	\begin{equation}\label{ppot}\tilde{Q}=Q-\frac s n f,\qquad (s\in\R).\end{equation}
	Indeed, if we denote
	by $Z_{n,sf}$ the partition function with respect to \eqref{ppot} then
	\begin{equation}\label{cgfdiff}F_{n,f}'(s)=\frac d {ds}\log Z_{n,sf}-n\sigma(f).\end{equation}

	\subsection{Class of outpost potentials}\label{secop} Let $Q$ be a potential obeying the conditions in Section \ref{cge} and assume that the droplet $S$ is \textit{connected}.
	We also require that $Q$ is real-analytic near the outer boundary
	$$C_1:=\d_* S.$$
	
	By Sakai's regularity theorem, $C_1$ is a finite union of regular, real-analytic arcs and possibly a finite number of singular points,
	see \cite{LM} or \cite[Section 3]{AC}.
	We assume that there are no singular points and thus that $C_1$ is
	an everywhere regular (real-analytic) Jordan curve.

	Let $\phi_1:\Ext(C_1)\to \D_e$ be the unique conformal mapping of the form
	\begin{equation}\label{otf}\phi_1(z)=\frac 1 {r_1} z+a_{0,1}+a_{1,1}\frac 1 z+\cdots\end{equation}
	for $z$ in a neighbourhood of $\infty$, where $r_1=\Cap(C_1)>0$ is the capacity of $C_1$.

	Next fix any number $r_2>r_1$ and define an analytic Jordan curve $C_2$ by
	$$C_2=\{z\in\Ext(C_1)\,;\,|\phi_1(z)|=r_2/r_1\}.$$

	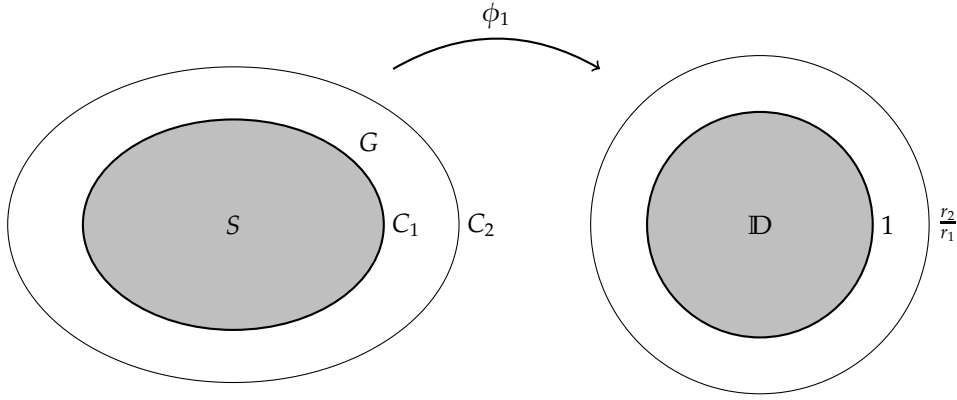
\begin{figure}
		\begin{tikzpicture}
			\filldraw[thick, fill=gray!50]  (0,0) ellipse (2cm and 1.4cm);
			\draw ellipse (3cm and 2.1cm);

			\filldraw[thick, fill=gray!50] (7,0) circle (1.5cm);
			\draw (7,0) circle (2.25cm);

			\tikzset{
				point/.style = {circle, fill, inner sep=1.5pt},
				func/.style = {->, thick}
			}
			\node at (0,0) {$S$};
			\node at (2.3,0) {$C_1$};
			\node at (3.3,0) {$C_2$};
			\node (x1) at (2,2) {};
			\node (x2) at (5,2) {};
			\node at (7,0) {$\mathbb{D}$};
			\node at (8.7,0) {$1$};
			\node at (9.5,0) {$\frac {r_2}{r_1}$};
			\node at (1.8,1.1) {$G$};
			
			\draw[func] (x1) to[bend left] node[above] {$\phi_1$} (x2);

		\end{tikzpicture}
		\caption{Coincidence set $S^*$ for an outpost potential}
		\label{fig1}
	\end{figure}

	It is easily seen that $C_2$ has capacity $r_2$ and that the normalized conformal map $\phi_2:\Ext(C_2)\to \D_e$ satisfies
	$$\phi_2(z)=\frac {r_1}{r_2}\phi_1(z).$$

	We assume that the coincidence set satisfies $S^*=S\cup C_2$ and we write $G$ for the \textit{gap} between $C_1$ and $C_2$, i.e., $G$ is the bounded, ring-shaped domain with boundary $C_1\cup C_2$.
	
	Denote by \begin{equation}\label{varpi}\varpi(z)=\frac{\log|\phi_1(z)|}{\log(r_2/r_1)}\end{equation} the harmonic function on $G$ whose boundary values are $0$ on $C_1$ and $1$ on $C_2$.
	
	By standard theory for the Dirichlet problem \cite{B,GM} there exists a holomorphic function $h_1(z)$ on $G$ and a real constant $c$ such that the harmonic function on $G$
	\begin{equation}\label{diric}H(z):=\re h_1(z)+c \varpi(z)\end{equation}
	has boundary values
	$$H(z)=\tfrac{1}{2}\log (\Delta Q(z)),\qquad (z\in C_1\cup C_2).$$

	We shall require the following \textit{compatibility condition}:
	
	\medskip

	(C) The function $h_1(z)$ extends to a (bounded) holomorphic function on $\Ext C_1$.
	
	\medskip
	
	Assuming condition (C), we make the decomposition \eqref{diric} unique by normalizing $h_1(z)$ to be real at $z=\infty$. We note that this $h_1(z)$ is also the unique normalized holomorphic function on $\Ext C_1$ with $\re h_1(z)=\tfrac{1}{2}\log (\Delta Q(z))$ on $C_1$.
	
	We also introduce
	\begin{equation}\label{h1h2}h_2(z):=h_1(z)+c\end{equation}
	and note that $h_2(z)$ is the unique normalized holomorphic function on $\Ext C_2$ which satisfies $\re h_2(z)=\frac{1}{2}\log(\Delta Q(z))$ on $C_2$.

	\medskip
	
	We refer to a potential satisfying the above assumptions (including the standing assumptions in Section \ref{cge}) as an \textit{outpost potential}; see Figure \ref{fig1}.

	\begin{rem} \label{rem1}
		In our analysis below, it is convenient (but not important) to impose the additional condition that the droplet $S$ be simply connected, i.e., $S$ equals to the closure of $\Int C_1$. This is tacitly assumed in our proofs in the succeeding sections.
	\end{rem}

	\subsubsection{Examples of outpost potentials} \label{ex1} To construct outpost potentials, we may start with an arbitrary
	quasi-harmonic potential \footnote{or Hele-Shaw potential} $Q_1$. This means that
	$$Q_1(z)=\Delta_1\cdot |z|^2+\calH(z)$$
	for all $z$ in a neighbourhood of the droplet $S:=S[Q_1]$ where $\Delta_1>0$ is an arbitrary but fixed number and $\calH(z)$ is harmonic in a neighbourhood of $S$.
	
	We assume in addition that $C_1:=\d_* S$ is a regular Jordan curve.

	Now consider the exterior conformal map $\phi_1:\Ext(C_1)\to\D_e$ of the form \eqref{otf} and define $C_2:=\phi_1^{-1}((\frac {r_2} {r_1})\cdot \T)$, where $r_2>r_1$ is arbitrary but fixed.

	Next consider the obstacle function $\check{Q}_1(z)$. Fix small neighbourhoods $N_1$ of $S$ and $N_2$ of $C_2$ whose closures are disjoint: $\overline{N_1}\cap \overline{N_2}=\emptyset$.

	Finally fix an arbitrary number $\Delta_2>0$ and define $Q(z)$ by the prescription
	$$Q(z):=\begin{cases}Q_1(z)&\text{if}\quad z\in \overline{N_1}\cr
		\check{Q}_1(z)+\dfrac {\Delta_2}{2|\phi_1(z)\phi_1'(z)|^2}\left[|\phi_1(z)|^2-(\frac {r_2}{r_1})^2 \right]^{\,2}&\text{if}\quad  z\in \overline{N_2}\cr
		+\infty &\text{otherwise}
	\end{cases}.$$

	It is readily checked that $Q(z)$ is an admissible potential with droplet $S$ and that $\check{Q}(z)=\check{Q}_1(z)$. Hence the coincidence set $S^*:=S^*[Q]$ satisfies $S^*=S\cup C_2$.

	We also see that $\Delta Q=\Delta_1$ on $C_1$ and by a straightforward computation we have $\Delta Q=\Delta_2$ on $C_2$.
	
	Hence the condition (C) holds with
	$$h_1(z)\equiv \tfrac{1}{2}\log \Delta_1,\qquad c=\tfrac{1}{2}\log\frac {\Delta_2}{\Delta_1}.$$
	
	\begin{rem}
		Complements of droplets of quasi-harmonic potentials form a very rich and well-studied class of sets, known as (generalized) quadrature domains. \footnote{It is standard to call them ``quadrature domains'' even if they are disconnected}
		See e.g. \cite{LM,GTV,S,TBAZW} and references therein.
	\end{rem}

	\subsection{Number of particles near the outpost} Let $Q(z)$ be an outpost potential as in Section \ref{secop}.
	
	We also fix a smooth, bounded, real-valued function $\wfun(z)$ on $\C$ with $\wfun=0$ in a neighbourhood of $S$ and
	$\wfun=1$ on a neighbourhood of $C_2$. \footnote{$\omega(z)$ is not to be confused with the harmonic measure $\varpi(z)$ for $G$}
	
	Write
	$$N_n[C_2]:=\fluct_n \wfun=\sum_{j=1}^n \wfun(z_j),$$
	which is essentially just
	the number of particles that fall near the outpost $C_2$.

	Recall from \cite{ACC2} that a random variable $X$ taking values in $\Z_+=\{0,1,2,\ldots\}$ is said to have a Heine distribution with parameters $(\theta,q)$, where $\theta>0$, $0<q<1$, if
	\begin{equation}\label{pdf}\Prob(\{X=j\})=\frac 1 {(-\theta;q)_\infty}\frac {q^{\frac 1 2 j(j-1)}\theta^j}{(q;q)_j},\qquad (j\in\Z_+),\end{equation}
	where the $q$-Pochhammer symbol $(z;q)_j$ is given by
	\begin{equation}\label{qpoc0}(z;q)_j=\prod_{i=0}^{j-1}(1-zq^i)\end{equation}
	and \begin{equation}\label{qpoc}(z;q)_\infty=\lim_{j\to\infty}(z;q)_j=\prod_{i=0}^\infty(1-zq^i).\end{equation}
	
	The fact that the Heine distribution is a probability distribution follows from the $q$-binomial theorem (see eg \cite{GR} or \cite{ACC2}).

	We write $X\sim \He(\theta,q)$ to denote that $X$ has the probability function \eqref{pdf}.
	
	It is convenient to recall the following elementary fact; we refer to the proof of \cite[Corollary 1.3]{ACC2} for a derivation.

	\begin{lem} \label{vwe} Suppose that $X\sim \He(\theta,q)$. The cumulant generating function $F_X(s)=\log\E\exp(sX)$ is given by
		\begin{equation}\label{cgf_he}F_X(s)=\log[(-\theta e^s;q)_\infty]-\log[(-\theta;q)_\infty].\end{equation}
	\end{lem}

	We have the following theorem, which generalizes \cite[Corollary 1.10]{ACC2} beyond the radially symmetric case.

	\begin{thm} \label{main1} Assume an outpost potential as in Section \ref{secop}. Then the random variables $N_n[C_2]$ converge in distribution as $n\to\infty$ to a Heine distribution with parameters $(\theta,q)$ where
		$$\theta=\frac {r_1} {r_2}e^{-c},\qquad q=\big(\frac {r_1} {r_2}\big)^2,$$
		where $r_k=\Cap C_k$ for $k=1,2$ and $c$ is the constant in \eqref{diric} appearing in the solution of the Dirichlet problem with boundary values $\frac 1 2\log \Delta Q$.

		More precisely, if $X\sim \He(\theta,q)$ and if $F_X(s)$ is as in \eqref{cgf_he}, then
		$$F_{n,\omega}(s)=F_X(s)+\bigO\biggl(\frac{\log^{5/2}n}{\sqrt{n}}\biggr),\qquad (n\to\infty),$$
		where $\delta_n$ is given by \eqref{deltan} and the $\bigO$-constant is uniform for $|s|\le \log n$.
	\end{thm}

	\begin{rem} For $X\sim \He(\theta,q)$ the expectation is $\E X=\sum_{j=0}^\infty \frac {\theta q^j}{1+\theta q^j}$; cf.~\cite[Proposition 1.11]{ACC2}.
	\end{rem}

	\begin{rem} In the theory of Hermitian random matrices, outposts in the form of a small $n$-dependent interval which shrinks to a point has been studied in the works \cite{A97,BL,Cl,Mo}. This is also known as the ``birth of a cut''. It is shown in \cite{BL,Cl,Mo} that if one rescales suitably about the outpost, then as $n\to\infty$,
		the process of eigenvalues near the outpost converges to GUE$(k)$ where $k$ is deterministic, and depends on the rate at which the outpost contracts to a point. (If the interval degenerates to a singleton then $k=0$, i.e., there are no eigenvalues near the outpost.)
		
		Clearly, this $1$-dimensional situation is very different from ours, where the outpost is a Jordan curve which attracts a random, Heine distributed number of particles. We do not have a deep explanation, but in some heuristic sense it is not unreasonable that the cases should be different. Indeed, the particles near the outpost are very distant from each other in our two-dimensional setting, but this spatial independence is of course lost in the situation of the birth of a cut.

		We expect (see a conjecture in \cite{ACC2}) that an outpost of capacity zero (in particular, a singleton set) should attract no particles, i.e., the number of particles near the outpost converges to zero in distribution, as $n\to\infty$.
		This hypothesis is supported by computations in the radially symmetric case in \cite[Theorem 1.12]{ACC2}, as well as by results in dimension 1 (already mentioned above).
		
	\end{rem}
	
	\begin{rem} For a potential as in Example \ref{ex1}, i.e., $\Delta Q\equiv \Delta_k$ on $C_k$ for $k=1,2$, the parameter $\theta$ in Theorem \ref{main1} becomes $$\theta=\tfrac {r_1} {r_2}\sqrt{\tfrac {\Delta_1}{\Delta_2}}.$$
		In particular, we recover the result on radially symmetric potentials in \cite[Corollary 1.10]{ACC2} as a special case of Theorem \ref{main1}.
	\end{rem}

	\subsection{Class of spectral gap potentials}  \label{csg} We now
	fix a parameter $\tau_*$ with $0<\tau_*<1$.
	
	Consider a potential $Q(z)$ such that the droplet $S$ consists of two connected components
	$$S=S_{\tau_*}\cup (S\setminus S_{\tau_*})$$ where $S_{\tau_*}=S[Q/\tau_*]$ is the droplet of mass
	$\tau_*$, i.e., $\sigma(S_{\tau_*})=\tau_*$. We also assume that the droplet equals to the coincidence set: $S=S^*$.

	It is convenient (though not essential) to assume that $S_{\tau_*}$ is simply connected. In addition we assume that the outer boundary $\d_* S$ is an everywhere smooth Jordan curve.

	Let us write $G$ for the bounded component of $\C\setminus S$; we refer to $G$ as the ``spectral gap'', cf.~Figure \ref{fig2}.
	
	We assume that $Q(z)$ is real-analytic in a neighbourhood of the boundary of $G$,
	$$\d G=C_1\cup C_2,$$
	where $C_{1}$ and $C_{2}$ denote the inner boundary and the outer boundary of $G$ respectively. We assume throughout that these curves are everywhere regular (real-analytic).

	As before we write
	\begin{align}\phi_{1}(z)=\frac 1 {r_{1}}z+a_{1,0}+a_{1,1}\frac 1 z+\cdots\end{align}
	for the exterior conformal map $\Ext(C_{1})\to\D_e$ where $r_{1}=\Cap C_{1}$.
	(To emphasize the dependency on $\tau_*$ we sometimes write $\phi_{1,\tau_*}$ and $r_{1,\tau_*}$.)

	We also assume that $C_{2}$ is a level curve for $|\phi_{1}|$, i.e., that there is a number $r_2=r_{2,\tau_*}>r_{1}$ such that
	\begin{align}\label{a1}C_{2}=\phi_{1}^{-1}\biggl(\bigl(\frac {r_{2}} {r_{1}}\bigr)\cdot \T\biggr).\end{align}
	
	Writing $\phi_2=\phi_{2,\tau_*}$ for the normalized conformal map $\Ext(C_{2})\to\D_e$, we then have
	$$\phi_{2}(z)=\frac {r_{1}} {r_{2}} \phi_{1}(z)=\frac 1 {r_{2}}z+a_{2,0}+a_{2,1}\frac 1 z+\cdots.$$
	
	We finally require a compatibility condition for the spectral gap $G$, having  the same appearance as for outpost potentials. For ease of reference, we spell it out again:
	
	\medskip
	
	(C) If $H(z)$ is the harmonic function on $G$ with boundary values $\frac{1}{2}\log \Delta Q$ and if \begin{align}\label{dir2}H(z)=\re h_1(z)+c\varpi(z)\end{align} where $h_1$ is holomorphic in $G$ and $\varpi(z)$ is the harmonic function on $G$ with boundary values $\varpi=0$ on $C_1$ and $\varpi=1$ on $C_2$, then $h_1$ extends analytically to $\Ext C_1$. We normalize $h_1$ so that $h_1(\infty)$ is real, and we write $h_2:=h_1+c$.
	
	\begin{figure}
		\begin{tikzpicture}
			\filldraw[thick, fill=gray!50]  (0,0) ellipse (3.8cm and 2.3cm);
			\filldraw[thick, fill=white!50]  (0,0) ellipse (3cm and 2.1cm);
			\filldraw[thick, fill=gray!50]  (0,0) ellipse (2cm and 1.4cm);

			\filldraw[thick, fill=gray!50] (8,0) ellipse (2.6cm and 2.35cm);
			\filldraw[thick, fill=white!50] (8,0) circle (2.25cm);
			\filldraw[thick, fill=gray!50] (8,0) circle (1.5cm);

			\tikzset{
				point/.style = {circle, fill, inner sep=1.5pt},
				func/.style = {->, thick}
			}
			\node at (0,0) {$S_{\tau_*}$};
			\node at (1.7,0) {$C_1$};
			\node at (2.7,0) {$C_2$};
			\node at (1.8,1.1) {$G$};
			\node at (3,-1.9) {$S\setminus S_{\tau_*}$};
			\node (x1) at (2,2) {};
			\node (x2) at (5,2) {};
			\node at (8,0) {$\mathbb{D}$};
			\node at (9.2,0) {$1$};
			\node at (10,0) {$\frac {r_{2}}{r_{1}}$};

			\draw[func] (x1) to[bend left] node[above] {$\phi_1$} (x2);

		\end{tikzpicture}
		\caption{Droplet of a spectral gap potential}
		\label{fig2}
	\end{figure}
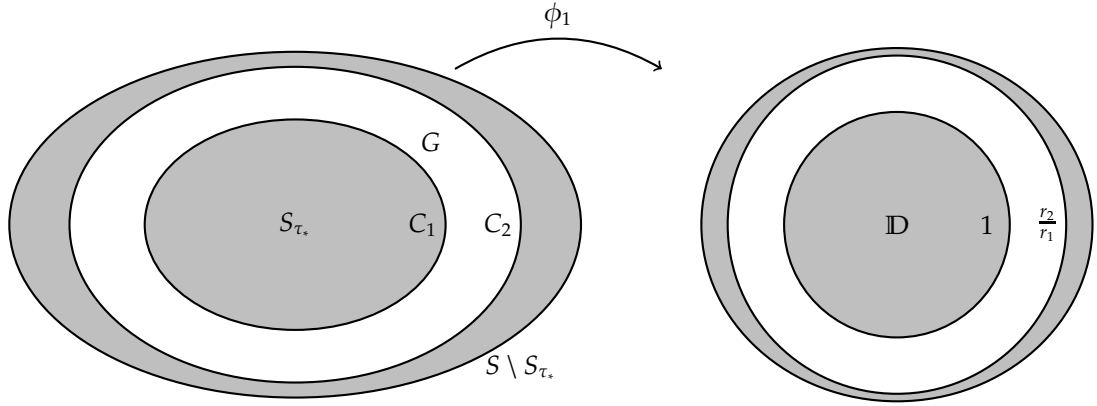

	\medskip

	In what follows, we refer to a potential meeting all the above requirements as a \textit{spectral gap potential}.

	\begin{rem} The main geometrical condition on the spectral gap is that $C_2$ is a level curve for $|\phi_1|$. We do not see a reason why our results should be true for other $C_2$.
	\end{rem}

	\subsubsection{Examples of spectral gap potentials} Let $Q_0(z)$ be any outpost potential meeting the requirements in Section \ref{secop}; for example
	$Q_0(z)$ may be as in Section \ref{ex1}.

	We form a new potential $Q$ by setting $Q=Q_0/\tau$ where the constant $\tau>1$ is close enough to $1$ to ensure that the droplet $S=S[Q]$ has two connected components and the boundary
	$\d S$ consists of everywhere regular Jordan curves. Setting $\tau_*=1/\tau$ it is now easy to see that $S$ is decomposed in the connected components
	$S=S_{\tau_*}\cup (S\setminus S_{\tau_*})$. The conditions \eqref{a1} and (C) are automatic.

	\subsection{Number of particles near the ring $S\setminus S_{\tau_*}$} Fix a smooth test function $\wfun(z)$ with $\wfun=0$ in a neighbourhood of $S_{\tau_*}$ and $\wfun=1$ in a neighbourhood
	of $S\setminus S_{\tau_*}$.

	The random variable
	$$\trace_n\omega=\sum_{j=1}^n \wfun(z_j)$$ can be interpreted as number of particles that fall near the ring-shaped component $S\setminus S_{\tau_*}$ of the droplet.

	Since $\sigma(S_{\tau_*})=\tau_*$ we have $\sigma(S\setminus S_{\tau_*})=1-\tau_*$ and thus
	$$\fluct_n \wfun=\sum_{j=1}^n \wfun(z_j)-n(1-\tau_*).$$

	The cumulant generating function $F_{n,\wfun}(s)$ of $\fluct_n \wfun$ satisfies
	\begin{equation}\label{cgf_ring}F_{n,\wfun}'(s)=\frac d {ds}\log Z_{n,s\wfun}-n(1-\tau_*).\end{equation}

	Define a number $x_n$ in the interval $0\le x_n<1$ to be the fractional part of $n\tau_*$, i.e.,
	$$x_n=\{ n\tau_*\}:=n\tau_*-\lfloor n\tau_*\rfloor.$$

	We have the following result.

	\begin{thm}\label{mth2} Assume a spectral gap potential as in Section \ref{csg}. Then for large $n$, the distribution of the random variable $\fluct_n \wfun$ is well approximated by the difference $X_n^+-X_n^-$ where $X_n^+,X_n^-$ are independent Heine distributed random variables,
		$$X_n^+\sim\He\biggl( e^{-c}\bigl(\frac {r_{1}}{r_{2}}\bigr)^{1+2x_n};\bigl(\frac {r_{1}}{r_{2}}\bigr)^2\biggr),\qquad X_n^-\sim\He\biggl(e^{c}\bigl(\frac {r_{1}}{r_{2}}\bigr)^{1-2x_n};\bigl(\frac {r_{1}}{r_{2}}\bigr)^2\biggr),$$
		where $c$ is the constant in \eqref{dir2}.

		More precisely we have the following convergence of cumulant generating functions
		$$F_{n,\wfun}(s)=F_{X_n^+}(s)+F_{-X_n^-}(s)+\bigO\biggl(\frac{\log^{5/2 }n}{\sqrt{n}}\biggr),\qquad (n\to\infty),$$
		where $\delta_n$ is given in \eqref{deltan} and the $\bigO$-constant is uniform for $|s|\le \log n$.
	\end{thm}

	\begin{rem} The difference $Y_n:=X_n^+-X_n^-$ can be expressed in terms of the discrete normal distribution. Indeed, we recall from \cite[Remark after Corollary 1.3]{ACC2}
		that
		\begin{equation}\label{dN}\Prob(\{Y_n=j\})=\frac 1 {Z_{\theta,q}} \theta_n^jq^{\frac 1 2 j(j-1)},\qquad (j\in\Z)\end{equation}
		where
		$$\theta_n=e^{-c}\bigl(\frac {r_1}{r_2}\bigr)^{1+2x_n},\qquad q=\bigl(\frac {r_1}{r_2}\bigr)^2,$$
		while $Z_{\theta,q}$ is the normalizing constant.
		A random variable having the probability function \eqref{dN} is said to have a discrete normal distribution.
	\end{rem}

	\begin{rem}
		The normal and the discrete normal distributions are characterized as the probability distributions on $\R$ and $\Z$ respectively which maximizes the entropy among all probability distributions with a given expectation and variance \cite{K}.
		
		However, the Heine distribution does not maximize the entropy among corresponding
		probability distributions on $\Z_+$.  Further the Heine distribution is not infinitely divisible.
		
		In a way, these properties of the Heine distribution reinforce the principle
		that convergence of fluctuations is very different from the convergence appearing in central limit theorems, e.g.~in \cite[Section 4]{GK}.
	\end{rem}

	\subsection{A convergence result for smooth linear statistics} We continue to consider a spectral gap potential $Q(z)$
	as in Section \ref{csg}.
	
	Write $C^\infty_b(\C)$ for the linear space of smooth, bounded, real-valued functions on $\C$.
	
	Consider the subspaces $\calG$ and $\calH$ where
	\begin{enumerate}[label=(\roman*)]
		\item \label{ett} $\calG$ consists of all functions $f(z)$ of the type
		$$f=f_0(z)+\re g(z),$$
		$g(z)\in C^\infty_b(\C)$ is holomorphic (and bounded) on each connected component of $\hat{\C}\setminus S$; the function $f_0(z)\in C^\infty_b(\C)$ satisfies
		$$f_0=0\qquad \text{along}\qquad \d S.$$
		\item \label{tva} $\calH$ consists of all constant multiples $\para\wfun(z)$ where $\para\in\R$ and $\wfun(z)$ is a fixed smooth, bounded function with $\wfun=0$ in a neighbourhood of $S_{\tau_*}$ and $\wfun=1$ in a neighbourhood of $S\setminus S_{\tau_*}$.
	\end{enumerate}

	We have the following lemma; a proof is given in Section \ref{decomp}.

	\begin{lem} \label{de1} Each $f\in C^\infty_b(\C)$ can be represented as $f=f_1+f_2$ where $f_1\in\calG$ and $f_2\in\calH$.
	\end{lem}

	The representation in Lemma \ref{de1} is not unique but, as we will see, it reduces the study of fluctuations to linear statistics of classes $\calG$ and $\calH$. To describe the fluctuations we need some preparation.
	
	First, for a continuous function $f(z)$ we define its Poisson modification $f^S(z)$. This is the continuous function on $\hat{\C}$ defined by
	\begin{enumerate}
		\item $f^S=f$ on $S$,
		\item $f^S$ is harmonic (and bounded) on each connected component of $\hat{\C}\setminus S$.
	\end{enumerate}
	
	Note that $f^S$ typically has a ``jump'' in the first derivatives across the boundary $\d S$.
	
	For $f_1=f_0+\re g$ as in \ref{ett}, we have
	\begin{equation}f_1^S(z):=\re g(z)+f_0(z)\cdot \1_S(z).\end{equation}
	
	If $G$ is the gap domain bounded by $C_1$ and $C_2$, then $f_2=\para\omega$ in \ref{tva} has
	\begin{equation}f_2^S(z)=\para\varpi(z),\qquad  (z\in G),
	\end{equation}
	while $f_2^S= 0$ on $\Int C_1$ and $=\para$ on $\Ext C_2$; $\varpi$ is the harmonic measure from \eqref{varpi}.

	\smallskip
	
	We shall find that fluctuations of linear statistics in the class $\calG$ have an asymptotic Gaussian distribution, generalizing the Gaussian field convergence in \cite{AM}. It is expedient to recall some details concerning this fundamental result.

	In the following we write $L$ for a fixed, smooth function supported in a small neighbourhood of $S$ with
	\begin{equation}\label{Ldef}L=\log\Delta Q\qquad\text{in a neighbourhood of } S.\end{equation}

	We next define the Neumann jump operator $\calN^S$ associated with the droplet $S$.
	
	\smallskip

	Given a smooth function $h$
	and a point $p\in \d S$ we write $\d_N h(p)$ for the directional derivative in the direction $N$ which is normal to $\d S$ at $p$ and points out from $S$.

	We define a function $h^S$ on $\C\setminus S$ as the bounded solution to the Dirichlet problem with boundary values $h^S=h$ on $\d S$.
	We note that $h^S$ restricted to $\C\setminus S$ is smooth up to the boundary (see \cite{B,GM} for details).

	The Neumann jump operator is now defined at the point $p$ by
	$$\calN^S(h)=-\d_N (h-h^S),$$
	where the normal derivative is taken in the exterior of the boundary.

	As before, the cumulant generating function is written
	$$F_{n,f}(t)=\log \E_n \exp(t\fluct_n f),\qquad (t\in \R).$$

	We have the following theorem, which we prove (in Section \ref{brief}) using a variant of the technique of limit Ward identities from \cite{AM} and \cite[Section 6]{ACC1}.

	\begin{thm} \label{m7} Let $f\in\calG$. Then there is a number $\beta>0$ such that
		$$F_{n,f}(t)=te_f+\frac {t^2} 2 v_f+\bigO(n^{-\beta})$$
		where, with $L(z)$ as in \eqref{Ldef},
		\begin{align}e_f&=\frac 1 2\int_Sf\cdot \Delta\log\Delta Q\, dA+\frac 1 {8\pi}\oint_{\d S}\d_n f\,ds+\frac 1 {8\pi}\oint_{\d S} f\cdot\calN^S(L)\,ds,\label{ef}\\
			v_f&=\frac 1 4 \int_\C |\nabla (f^S)|^2\, dA.\label{vf}
		\end{align}
		In particular $\fluct_n f$ converges in distribution as $n\to\infty$ to the normal $N(e_f,v_f)$.
	\end{thm}

	We have the following main result.

	\begin{thm} \label{main3} Assume a spectral gap potential as in Section \ref{csg}. Let $f=f_1+\para\wfun$ where $f_1\in\calG$, and let $Y\sim N(e_{f_1},v_{f_1})$ be a Gaussian random variable with cumulant generating function
		$$F_Y(t)=te_{f_1}+\frac {t^2} 2 v_{f_1}.$$
		Also let $X_n^+$, $X_n^-$ be Heine distributed random variables with parameters as in Theorem \ref{mth2}. Then we have the convergence of cumulant generating functions
		$$F_{n,f}(t)=F_{\para X_n^+}(t)+F_{-\para X_n^-}(t)+F_Y(t)+\bigO(n^{-\beta}),\qquad (n\to\infty),$$
		where $\beta>0$ is a constant and the convergence is uniform for $|t|\le \log n$.
	\end{thm}

	In particular, for $f\in \calG$ the fluctuations $\fluct_n f$ converge to a Gaussian field, and for $f\in\calH$ $\fluct_n f$ has an oscillatory discrete Gaussian distribution.

	Note that Theorem \ref{main3} does not follow from Theorem \ref{main1} and Theorem \ref{mth2}, since it claims an additional asymptotic independence among $X_n^+,X_n^-,Y$. However, with an additional twist, our proofs can be adapted to accommodate the required independence (cf.~Section \ref{Pmain3}).
	
	\subsection{Orthogonal polynomials in the bifurcation regime} Our analysis depends on an asymptotic, leading order formula for the monic orthogonal polynomials $\tilde{p}_{j,n}(z)=z^j+a_{n,j-1}z^{j-1}+\cdots +a_{n,0}$ with respect to the norm $\int_\C|f|^2 e^{-n\tilde{Q}}\, dA$. (See \eqref{ppot}.)
	
	When $j$ is far from the critical index $n\tau_*$ the weighted polynomial $\tilde{p}_{j,n}(z)e^{-n\tilde{Q}(z)/2}$ is essentially supported near the outer boundary of the $\tau$-droplet, where $\tau=j/n$; its leading order asymptotic behaviour is then known due to \cite{HW}, see also \cite[Section 5]{AC}. (Here, we put $\tau_*=1$ if $Q$ is an outpost potential.)

	The situation is very different if $j=n\tau_*+\bigO(\log^2 n)$.
	
	In this situation which we call the \textit{bifurcation regime} the weighted polynomial $\tilde{p}_{j,n}(z)e^{-n\tilde{Q}(z)/2}$ has two peaks along disjoint Jordan curves close to $C_1$ and $C_2$, respectively. The asymptotics is described in Theorem \ref{gondor}
	for $\tau_*=1$ and Theorem \ref{gondor2} for $\tau_*<1$.
	
	Using those results it is possible to work out asymptotics for $1$- and $2$-point correlation functions, in a way somewhat similar to \cite{ACC1,ACC0}. However, such an investigation would carry us too far afield.

\begin{rem} The bifurcation regime differs from the ``soft-edge regime'' studied in sources such as \cite{HW,ACC1}, which accommodates the well-known $\erfc$-asymptotic. The bifurcation regime is associated with $n$-dependent oscillations as in e.g.~\cite{ACC1}, and involves much fewer indices $j$ near the critical $n\tau_*$, i.e., just $\bigO(\log^2 n)$ indices. By contrast, the soft edge asymptotics studied in \cite{HW} involves $\bigO(\sqrt{n\log n})$ terms, of which $\bigO(\log^2 n)$ constitutes a negligible share. See also \cite{ACC0} for other sorts of edge conditions.
\end{rem}

	\subsection{A surmise on the large-$n$ expansion of the free energy} It is believed that there exists a large $n$-expansion of the free energy $\log Z_n$ of the form
	\begin{equation}\label{expo}\log Z_n=C_0n^2+C_1n\log n+C_2n+C_3\log n+C_4+o(1),\qquad
		(n\to\infty)\end{equation}
	where $C_0,\ldots,C_4$ are certain geometric functionals (e.g. $C_0=-I_Q[\sigma]$).

	In a remarkable paper \cite{ZW}, Zarbrodin and Wiegmann conjecture (``up to constants'') a formula for the constant term $C_4$, for potentials $Q$ such that the droplet $S$ is connected. Let us denote by $F(Q,S)$ this term.  The formula for $F(Q,S)$ is given in terms of certain spectral determinants and is called the Polyakov-Alvarez formula.
	
	Recently, in \cite{BKS}, the Polyakov-Alvarez formula was verified (with the correct constants) for radially symmetric and globally subharmonic potentials.
	
	The global subharmonicity condition on the potential is removed in \cite{ACC2} where disconnected droplets (finite union of concentric annuli) are studied in great details. To be precise, we assume in \cite{ACC2} that
	$$S=\bigcup_{\nu=0}^N S_\nu,\qquad S_\nu=\{a_\nu\le |z|\le b_\nu\}$$
	where $0\le a_0<b_0<\cdots<a_N<b_N$.
	
	The term $C_4=C_4(n)$ in \eqref{expo} varies with $n$, but is bounded as $n\to\infty$; in \cite{ACC2} we find that
	\begin{equation}\label{pred}C_4=F(Q,S)+\mathcal{G}_n,\end{equation}
	where $F(Q,S)=\sum_{\nu=0}^NF(Q,S_\nu)$ is the straightforward sum of ``Polyakov-Alvarez terms'' of the connected components of the droplet.
	
	The term $\mathcal{G}_n$ in \eqref{pred} measures the displacements of particles between the connected components, and
	is given in terms of $q$-Pochhammer symbols:
	\begin{align*}\mathcal{G}_n&=\sum_{\nu=0}^{N-1}(x_\nu\log\mu_\nu-x_\nu^2\log\rho_\nu)+\sum_{\nu=0}^{N-1}\log[(-\rho_\nu\mu_\nu;\rho_\nu^2)_\infty]+\sum_{\nu=0}^{N-1}\log[(-\rho_\nu \mu_\nu^{-1};\rho_\nu^2)_\infty].
	\end{align*}
	Here
	$$\rho_\nu=\frac {b_\nu} {a_{\nu+1}},\qquad \mu_\nu=\sqrt{\frac {\Delta Q(b_\nu)}{\Delta Q(a_{\nu+1})}}\,\rho_\nu^{2x_\nu},$$
	and
	\begin{equation}\label{xnu}x_\nu=x_{\nu,n}=\tau_*^{(\nu)} n-\lfloor \tau_*^{(\nu)} n\rfloor\end{equation}
	where $\tau_*^{(\nu)}=\sigma(\{|z|\le b_\nu\})$ for $\nu=0,\ldots,N-1$.
	
	(Alternatively, one can express $\calG_n$ in terms of the Jacobi theta function, instead of the $q$-Pochhammer symbols. See the last remark in \cite[Section 1.2]{ACC2}.)
	
	We now observe that the parameters $x_\nu,\rho_\nu,\mu_\nu$ all have natural intrinsic meanings, beyond the radially symmetric case.
	
	Indeed, let $G_\nu=\{b_\nu<|z|<a_{\nu+1}\}$ be the $\nu$th spectral gap where $0\le \nu\le N-1$.
	Write $C_{1,\nu}=\{|z|=b_\nu\}$, $C_{2,\nu}=\{|z|=a_{\nu+1}\}$ for the inner and outer boundary component of $G_\nu$ respectively.
	
	Then:
	\begin{itemize}
		\item $\rho_\nu$ is the conformal type of $G_\nu$, i.e., $G_\nu$ is biholomorphic to the annulus $\{\rho_\nu<|z|<1\}$,
		\item $\mu_\nu= e^{-c_\nu}\rho_\nu^{2x_\nu}$ where $c_\nu$ is such that the solution $H_\nu(z)$ to the Dirichlet problem on $G_\nu$ with boundary values $\frac{1}{2}\log \Delta Q$ can be written
		$$H_\nu =\re g_\nu + c_\nu\varpi_\nu,$$
		for an analytic function $g_{\nu}$ on $G_\nu$ and where $\varpi_\nu$ is the harmonic function on $G_\nu$ which is equal to zero on the inner boundary and one on the outer boundary of $G_\nu$,
		
		\item $x_\nu$ is given by \eqref{xnu} where
		$\tau_*^{(\nu)}=\sum_{j=0}^\nu\sigma(S_j)$.
	\end{itemize}
	
	With these conventions, it seems plausible that the prediction \eqref{pred} should hold for droplets $S=\cup_0^N S_\nu$ where the components $S_\nu$ are separated by suitable ring-shaped gaps $G_\nu$ (to be precise: each gap $G_\nu$ should obey the conditions in Section \ref{csg}).

	In the case of radially symmetric $Q$, the expansion \eqref{pred} is obtained by a detailed saddle-point analysis of the norms of monic orthogonal polynomials, in
	\cite{ACC2}.

	\subsection{Comments.}

	Disconnected droplets appear among the scale of ensembles
	with discrete rotational symmetry from \cite{BGM,BM,BY,DS}. For this class, the boundary of the droplet $S$ is a lemniscate.
	Unlike our situation of ring-shaped spectral gaps, the complement $\C\setminus S$ is connected, and
	this affects the Coulomb gas in an essential way. The reason is that under Laplacian growth, the droplet grows only near the outer boundary $\d_* (S^*)$ of the coincidence set.

	As far as we are aware, fluctuations of systems with discrete rotational symmetry have so far not been studied in the literature, and this could be potentially quite interesting. From experience with other models it is natural to conjecture that the number of particles near a given component of the droplet should be asymptotically periodic in $n$ since the equilibrium mass of each component is equal; see \cite{ACC1,CFWW} or \cite[Remark 1.5]{ACCL}.

	In addition to soft edge ensembles, a lot of work has been done on random normal matrix ensembles with other boundary conditions, such as ``hard edge'', ``soft/hard edge'', and many others, see \cite{BF} as well as the introduction to \cite{ACC0}. In the regime of disconnected droplets, fluctuations with respect to a class of radially symmetric potentials with hard edge boundary conditions are studied in the works \cite{C,ACCL,Berezin}.
	The hard edge theory is parallel and is equally interesting.
	
The recent works \cite{BY2,CK,KKL} give examples of disconnected droplets created by inserting a point charge. In particular \cite{BY2} provides an example without rotational symmetry.
The note \cite{AL} studies rotationally symmetric situations where the Laplacian vanishes along a circle.

	As far as we are aware, the Heine distribution was introduced in the random matrix context only recently, in our paper \cite{ACC2}. However, the Heine distribution appears (in a very different way) in connection with free fermions, cf.~ \cite{BB,Bo}. (We thank Y.~
	Liao for this information.)

	The survey \cite{F} overviews fluctuation theorems, also in the context of beta-ensembles.
	We recall that the simultaneous works \cite{BBNY2,LSe} propose generalizations of the fluctuation result obtained for eigenvalues of random normal matrices in \cite{AM} (cf \cite{RV} for the Ginibre ensemble) to beta-ensembles. However our present results do not seem to be known in the literature on beta-ensembles. See also the concluding remarks in \cite{AT} for some related comments.

	\subsection{Plan of this paper} In Section \ref{aaf} we provide an approximation formula for monic orthogonal polynomials, which incorporates a bifurcation when the degree of the polynomial is close to the critical value ($n$ or $n\tau_*$). For these degrees, the polynomials have two distinct peaks, near the curves $C_1$ and $C_2$ respectively.

	In Section \ref{bifur1} we provide asymptotic formulas for the squared norms of the weighted orthogonal polynomials for an outpost ensemble. We also prove a theorem on pointwise asymptotics for orthonormal polynomials in the bifurcation regime, see Theorem \ref{gondor}. Again these display a bifurcation as the degree passes the critical value.

	In Section \ref{proof: outpost} we prove Theorem \ref{main1} on the number of particles near an outpost.

	In Section \ref{bifur2}, we give our main result for norms of monic orthogonal polynomials in spectral gap ensembles, Theorem \ref{insert2}. We also prove Theorem \ref{gondor2} on pointwise asymptotics for the orthonormal polynomials in the bifurcation regime.

	In Section \ref{numgap} we prove Theorem \ref{mth2} on the number of particles near the component $S\setminus S_{\tau_*}$.

	In Section \ref{decomp} we prove the decomposition Lemma \ref{de1} for smooth functions.
	
	In Section \ref{brief} we use this to prove Theorem \ref{m7} on Gaussian fluctuations, by adapting the method using Ward's
	identity from \cite{AM,ACC1}.
	
	In Section \ref{Pmain3} we prove Theorem \ref{main3}.

	\subsubsection*{Further notation} In what follows, the notation $A_n\lesssim B_n$ means that there is $n_0$ such that $A_n\le CB_n$ where $C$ is some constant independent of $n\ge n_0$.

	Exterior discs are denoted $\D_e(\rho)=\{|w|>\rho\}\cup\{\infty\}$.

	For a measurable ``weight function'' $\phi(z)$, we define the $\phi$-norm by $$\|f\|_\phi^2:=\int_\C|f|^2e^{-\phi}\, dA$$ and we write $L^2_\phi$ for the corresponding $L^2$-space of complex-valued, measurable functions. Also we let $(\cdot, \cdot)_{\phi}$ denote the corresponding inner product.

    It is also convenient to use the abbreviation	\begin{equation}\label{deltan}\delta_n:=\sqrt{\tfrac{\log n} n}.\end{equation}

	\section{An approximation formula for monic orthogonal polynomials} \label{aaf}
	In this section we introduce an approximation formula for orthogonal polynomials in the ``bifurcation regime'' where the methods from \cite{HW} are not immediately applicable and need to be re-worked.
	
	We also discuss some related integration techniques, which are used in subsequent sections to prove that our approximation is ``good''.
	
	\smallskip
	
	Assume that $Q$ is an outpost potential satisfying the conditions in Section \ref{secop} and fix
	a smooth bounded function $\wfun(z)$ with $\wfun=0$ in a neighbourhood of the droplet $S$ and $\wfun=1$ in a neighbourhood of the outpost $C_2$. Also fix a real parameter $s$ and set
	\begin{equation}\label{wfunn}\tilde{Q}(z)=Q(z)-\frac s n \wfun(z).\end{equation}

	We shall give an approximation $\Phi_{j,n}(z)$ for the monic orthogonal polynomial $\tilde{p}_{j,n}(z)$ of degree $j$ with respect to the weight function $e^{-n\tilde{Q}}$ on $\C$.

	\smallskip

	\smallskip
	
	We remark that
	once the case of outposts is understood, the corresponding analysis of a spectral gap potential will be relatively simple (cf.~ Section \ref{bifur2} below).

	\subsection{Perturbed potentials and Hilbert spaces} Let $Q$ be an outpost potential as in Section \ref{secop} and form the perturbed potential $\tilde{Q}$ as in \eqref{wfunn}.

	Consider the space $L^2_{n\tilde{Q}}$ of measurable functions $f$ of finite squared norm
	$$\|f\|_{n\tilde{Q}}^2:=\int_\C|f|^2 e^{-n\tilde{Q}}\, dA.$$

	We want to estimate the squared norm
	\begin{equation}\label{genn} \tilde{h}_{j,n}(s):=\|\tilde{p}_{j,n}\|_{n\tilde{Q}}^2\end{equation}
	where $$\tilde{p}_{j,n}(z)=z^j+a_{j-1,n}z^{j-1}+\cdots+a_{0,n}$$ is the monic orthogonal polynomial of degree $j$ with respect to the norm in $L^2_{n\tilde{Q}}$.

	By Andr\'{e}ief's identity, the free energy of the ensemble in potential $\tilde{Q}$ is given by
	$$\log Z_{n,s\wfun}=\log (n!)+\sum_{j=0}^{n-1}\log \tilde{h}_{j,n}(s).$$

	In view of the relation \eqref{cgfdiff}, the function $F_n(s):=\log \E_n \exp (s\fluct_n \wfun)$ satisfies
	$$F_n(s)=\sum_{j=0}^{n-1}\log \frac {\tilde{h}_{j,n}(s)}{\tilde{h}_{j,n}(0)}.$$

	\subsection{Monic quasi-polynomials} \label{mqp} For $n$ large and
	$$j=n+\bigO(\log^2 n)$$
	we now define a ``monic quasi-polynomial'' $\Phi_{j,n}(z)$, which will subsequently be shown to approximate $p_{j,n}(z)$ well in an appropriate, weighted sense.

	To define $\Phi_{j,n}(z)$, we first introduce two auxiliary functions $q_1(z)$ and $q_2(z)$.

	For $k=1,2$, $q_k(z)$ is the bounded holomorphic function on $\Ext C_k$ which solves the Dirichlet problem
	$$\re q_k=Q\qquad\text{on}\qquad  C_k,$$
	and $\im q_k(\infty)=0$. Note that $q_k(z)$ continues analytically across $C_k$ to a neighbourhood of $\Ext(C_k)\cup C_k$.

	Our starting point is the following lemma, where we recall that $\phi_1$ and $\phi_2$ are the exterior conformal maps in Section \ref{secop}.

	\begin{lem} \label{lura} We have the identity
		\begin{equation}\label{lur}\phi_1(z)e^{\frac 1 2q_1(z)}=\phi_2(z)e^{\frac 1 2 q_2(z)}\end{equation}
		for all $z$ in a neighbourhood of infinity; in particular
		\begin{equation}\label{zick}\frac 1 {r_1}e^{\frac 1 2q_1(\infty)}=\frac 1 {r_2}e^{\frac 1 2q_2(\infty)}.\end{equation}
	\end{lem}

	\begin{proof} We first observe the identity
		\begin{equation}\label{idem}\check{Q}(z)=2\log|\phi_1(z)|+\re q_1(z)=2\log|\phi_2(z)|+\re q_2(z)\end{equation}
		for all $z$ in the exterior of $C_2$.

		To verify \eqref{idem}, we note that the function $2\log|\phi_1(z)|+\re q_1(z)$ agrees with $Q$ on $C_1$ and grows like $2\log|z|+\calO(1)$ as $z\to\infty$.
		The identity $\check{Q}(z)=2\log|\phi_1(z)|+\re q_1(z)$ for $z\in\Ext(C_1)$ hence follows by a suitable version of the maximum principle. A similar argument shows that
		$\check{Q}(z)=2\log|\phi_2(z)|+\re q_2(z)$ for $z\in\Ext(C_2)$. This shows \eqref{idem}.

		It follows from \eqref{idem} that the quotient
		$$\frac {\phi_1(z)e^{\frac 1 2q_1(z)}} {\phi_2(z)e^{\frac 1 2q_2(z)}}$$
		has absolute value $1$ on $\Ext C_2$ and thus equals to a unimodular constant. Since the quotient is real at infinity, the constant has to be one. This proves \eqref{lur} and \eqref{zick} follows
		by taking $z$ to $\infty$.
	\end{proof}

	Note that the left hand side of \eqref{lur} gives an analytic continuation of the right hand side, which is apriori just defined in $\Ext(C_2)$.

	\smallskip

	Now fix a positive integer $j$ and recall from \eqref{diric} that $h_1(z)$ denotes the normalized holomorphic function on $\Ext C_1$ with boundary values $\log\Delta Q$ on $C_1$.
	
	We consider the function in $\Ext(C_1)$ defined by
	\begin{align}\label{pjn}\Phi_{j,n}(z)=\frac {r_1^{j+1/2}} {e^{nq_1(\infty)/2}e^{h_1(\infty)/2}}\sqrt{\phi_1'(z)}\phi_1(z)^je^{nq_1(z)/2}e^{ h_1(z)/2},\end{align}
	which satisfies $\Phi_{j,n}(z)=z^j+O(z^{j-1})$ as $z\to\infty$.
	
	Note that the function $\Phi_{j,n}(z)$ is analytic in $(\Ext C_1)\setminus \{\infty\}$ and continues analytically across $C_1$ to a neighbourhood of $(\Ext(C_1)\cup C_1)\setminus\{\infty\}$.
	Therefore $\Phi_{j,n}(z)$ is unambiguously defined in such a neighbourhood.

	Using Lemma \ref{lura} and recalling the function $h_2=h_1+c$ from \eqref{h1h2} we have also (if $z\in\Ext(C_2)$)
	\begin{equation}\label{altpjn}\Phi_{j,n}(z)=\frac {r_2^{j+1/2}} {e^{nq_2(\infty)/2}e^{h_2(\infty)/2}}\sqrt{\phi_2'(z)}\phi_2(z)^je^{nq_2(z)/2}e^{ h_2(z)/2}.\end{equation}

	We next modify $\Phi_{j,n}(z)$ on a compact subset of $\Int (C_1)$ to a smooth function on $\C$. More precisely, we replace $\Phi_{j,n}$ by
	$\chi\Phi_{j,n}$ (defined to be zero where $\chi$ vanishes), where the smooth function $\chi$ is defined as follows.
	
	Fix two compact sets $K,K'\subset \Int(C_1)$ with $K$ contained in the interior of $K'$.
	We assume that $K$ is large enough that $\phi_1$ is well defined with $\phi_1'\ne 0$ on the closure of $\C\setminus K$. We also assume that $\Int(C_1)\setminus K$ is contained in the droplet $S$ and that
	the Laplacian $\Delta Q$ is bounded below by a positive constant on $\Int(C_1)\setminus K$.
	We fix a $\chi$ with $\chi=0$ on $K$ and $\chi=1$ on $\C\setminus K'$.

	\smallskip

	The following lemma will be used to show that for $j=n+\bigO(\log^2 n)$, the $L^2_{n\tilde{Q}}$-norm of $\chi \Phi_{j,n}$ is essentially concentrated in the union of the rings
	$$\calB_k=\{z\in\C\,;\, \dist(z,C_k)\le M\delta_n\},\qquad (k=1,2),$$
	where $M$ is a large enough constant and $\delta_n=\sqrt{\log n/n}$, see \eqref{deltan}.

	\begin{lem} \label{negl0} For $k=1,2$ write
		$$\calA_{k,n}=\{z\,;\, M\delta_n\le \dist(z,C_k)\le \delta\},$$
		where $\delta>0$ is fixed. Suppose that the integer $j$ is close to $n$ in the sense that
		\begin{equation}\label{bost}j=n+\bigO(\log^2 n),\qquad (n\to\infty).\end{equation}
		Then, given any $N>0$, we have by choosing $M$ large enough and $\delta>0$ small enough
		$$\int_{\calA_{k,n}}|\phi_k|^{2j}e^{n\re q_k}e^{\re h_k}e^{-nQ}|\phi_k'|\, dA=\bigO(n^{-N}),\qquad (n\to\infty).$$
	\end{lem}

	\begin{proof}

		In the set $\calA_{k,n}$ we have the identity
		\begin{equation}\label{weh}V=2\log|\phi_1|+\re q_1\end{equation}
		where $V$ is the harmonic continuation of $\check{Q}|_{\Ext(C_1)}$ inwards across $C_1$. By choosing $\delta>0$ small enough, this continuation is defined in $\calA_{k,n}$.

		We thus have
		\begin{equation}\label{inert}\int_{\calA_{1,n}}|\phi_1|^{2j}e^{\re h_k}e^{n\re q_1}e^{-nQ}|\phi_1'|\, dA=\int_{\calA_{1,n}}|\phi_1|^{2(j-n)}e^{\re h_k}|\phi_1'|e^{n(V-Q)}\, dA.\end{equation}

		A Taylor expansion as in \cite[Lemma 3.5]{AC} shows that, by perhaps taking $\delta>0$ somewhat smaller, there is a constant $c>0$ such that for all $z\in\calA_{1,n}$,
		\begin{equation}\label{tay1}(Q-V)(z)\ge c(\dist(z,C_1))^2.\end{equation}

		Further, if $\delta$ is chosen small enough that $\calA_{1,n}\subset \C\setminus K'$, then by \eqref{bost} we have the estimate
		\begin{equation}\label{jlarge} e^{-B\log^2 n \dist(z,C_1)}\le |\phi_1(z)|^{2(j-n)}\le e^{B \log^2 n \dist(z,C_1)} \end{equation}
		for a large enough constant $B$.

		Inserting the estimates \eqref{tay1}, \eqref{jlarge} in the right hand side of \eqref{inert}, using also that $|\phi_1'|$ and $e^{\re h_k}$ are bounded above and below on $\calA_{1,n}$, we find, if necessary by choosing $c>0$ somewhat smaller
		$$\int_{\calA_{1,n}}|\phi_1|^{2(j-n)}e^{\re h_k}|\phi_1'|e^{n(V-Q)}\, dA\lesssim \int_{M\delta_n}^\delta e^{-cnu^2}\, du\lesssim n^{-N},$$
		where the last estimate holds provided that the constant $M=M(c,N)$ is chosen large enough.

		A similar estimate for $\calA_{2,n}$ (omitted here) finishes the proof.
	\end{proof}

	\subsection{Integration in flow-coordinates}
	
	We now recall the foliation flow formalism from \cite{HW} (cf.~\cite[Section 5]{AC}).

	In what follows we let $V(z)$ be the harmonic continuation of $\check{Q}$ from $\Ext(C_1)$ inwards across $C_1$.

	Let $N_1$ be the $\delta$-neighbourhood of $C_1$ and $N_2$ the $\delta$-neighbourhood of $C_2$, where $\delta>0$ is small enough.

	Given a ``small'' parameter $t\in\R$
	we set
	$$\Gamma_{1,t}=\{z\in N_1\,;\, (Q-V)(z)=t^2\},\qquad \Gamma_{2,t}=\{z\in N_2\,;\, (Q-\check{Q})(z)=t^2\}.$$

	Then for small $t\ne 0$, $\Gamma_{1,t}$ is the union of two analytic Jordan curves $C_{1,t}^-$ (inside $C_1$) and $C_{1,t}^+$ (outside $C_1$). We set $C_{1,t}=C_{1,t}^-$ if $t< 0$ and
	$C_{1,t}=C_{1,t}^+$ if $t> 0$. We also set $C_{1,0}=C_1$.

	Likewise $\Gamma_{2,t}$ is the union of $C_{2,t}^-$ (inside $C_2$) and
	$C_{2,t}^+$ (outside $C_2$) and we set $C_{2,t}=C_{2,t}^-$ if $t< 0$ and
	$C_{2,t}=C_{2,t}^+$ if $t> 0$.

	For fixed $t$ with $|t|$ small we write $$U_{k,t}:=\Ext(C_{k,t})$$ for $k=1,2$ and write $\psi_{k,t}$ for the normalized \footnote{I.e., $\psi_{k,t}(\infty)=\infty$ and $\psi_{k,t}'(\infty)>0$}
	conformal mapping
	$$\psi_{k,t}:\D_e\to \phi_k(U_{k,t}).$$

	The mappings $\psi_{k,t}$ continue analytically across $\T$ for $k=1,2$ and
	$$(Q-V)\circ \phi_k^{-1}\circ\psi_{k,t}\equiv t^2\qquad \text{on}\qquad \T.$$

	Recall that $\delta_n=\sqrt{\frac {\log n} n}$ and define a neighbourhood $D_n^k$ of $\T$ (for $k=1,2$) by
	\begin{equation}\label{dnk}D_n^k=\bigcup_{-2M\delta_n\le t\le 2M\delta_n}\psi_{k,t}(\T),\end{equation}
	where $M$ is a suitable constant.

	Put
	$$\tilde{D}_n=\{(t,w)\,;\, w\in\T,\, -2M\delta_n\le t\le 2M\delta_n\}$$
	and define the flow maps $\Psi_k:\tilde{D}_n\to D_n^k$ by
	$$\Psi_k(t,w):=\psi_{k,t}(w).$$

	By \cite[Lemma 5.3]{AC}, the Jacobian of the mapping $\Psi_k$ satisfies
	\begin{equation}\label{jac}J_{\Psi_k}(t,w)=\bigl(\frac {|\phi_k'|}{\sqrt{2\Delta Q}}\bigr)\circ \phi_k^{-1}\cdot(1+\bigO(t)),\qquad (t\to 0).\end{equation}

	For $k=1,2$ and $f$ a suitable function (defined on $\phi_k(\C\setminus K)$) we define a new function (on $\C\setminus K$) by
	\begin{equation}\label{posit}\Lambda_{k,j,n}[f]:=\phi_k'\cdot \phi_k^j\cdot e^{\frac 1 2 nq_k}\cdot (f\circ \phi_k).\end{equation}

	Now write $\tau=j/n$, where $\tau$ is ``close'' to $1$.
	
	For $\rho<1$ close enough to $1$ we define a function $R_{k,\tau}$ on the exterior disc $\D_e(\rho)=\{|w|>\rho\}$ by
	\begin{equation}\label{ridg}R_{k,\tau}(z) =(Q-2\tau\log|\phi_k|-\re q_k) \circ \phi_k^{-1}(z).\end{equation}

	The map $\Lambda_{k,j,n}$ is then unitary
	$L^2_{nR_{k,\tau}}(\D_e(\rho))\to L^2_{nQ}(\phi_k^{-1}(\D_e(\rho)))$, i.e.,
	$$\int_{\phi_k^{-1}(\D_e(\rho))}\Lambda_{k,j,n}[f]\overline{\Lambda_{k,j,n}[g]}\, e^{-nQ}\, dA=\int_{\D_e(\rho)}f\bar{g}\, e^{-nR_{k,\tau}}\, dA.$$
	
	We also note that $\Lambda_{k,j,n}$ preserves holomorphicity, i.e., it restricts to a unitary map between Bergman spaces.

	\section{Bifurcation of orthogonal polynomials at an outpost}\label{bifur1}
	In this section, we prove an estimate for the squared norm $$\tilde{h}_{j,n}(s)=\|\tilde{p}_{j,n}\|_{n\tilde{Q}}^2$$ where $\tilde{p}_{j,n}$ is the order $j$ monic orthogonal polynomial in $L^2_{n\tilde{Q}}$.
	
	Our estimate is obtained by substituting the approximation $\chi \Phi_{j,n}$ for $\tilde{p}_{j,n}$ when $j=n+\bigO(\log^2 n)$; it is shown subsequently, in Subsection \ref{alo}, that $\log(\tilde{h}_{j,n}(s)/\tilde{h}_{j,n}(0))$ is negligible for $j\le n-\log^2 n$.

	\subsection{Approximate norm} Let $j$ be an integer which is close to $n$ in the sense
	$$j=n+\bigO(\log^2 n),\qquad (n\to\infty).$$

	We shall estimate the $L^2_{n\tilde{Q}}$-norm of $\chi\Phi_{j,n}$. To this end, we start by estimating the essential contribution, coming from integration over the rings $D_n^1$ and $D_n^2$ in \eqref{dnk}. When this is done, we will prove that
	the remaining integral is negligible.

	For $k=1,2$ we define a function $f_k$ on $\phi_k(\C\setminus K)$ by
	$\Lambda_{k,j,n}[f_k]=\Phi_{j,n}$, i.e., (by \eqref{pjn}, \eqref{altpjn})
	\begin{equation}\label{efk}f_{k}\circ \phi_k= r_k^{j+1/2}e^{-nq_k(\infty)/2} e^{-h_k(\infty)/2}(\phi_k')^{-1/2} e^{h_k/2}.\end{equation}

	Thus with $R_{k,\tau}(z)$ as in \eqref{ridg}
	\begin{align*}\int_{\phi_1^{-1}(D_n^1)\cup \phi_2^{-1}(D_n^2)}&|\chi\Phi_{j,n}|^2e^{-n\tilde{Q}}\, dA\\
		&=
		\sum_{k=1,2} e^{s(k-1)}r_k^{2j+1}e^{-nq_k(\infty)}e^{-h_k(\infty)} \int_{D_n^k}(|\phi_k'|^{-1} e^{\re h_k})\circ \phi_k^{-1} \cdot e^{-nR_{k,\tau}}\, dA.\end{align*}

	Note that $R_{k,1}\circ\psi_{k,t}\equiv t^2$ on $\T$ and (cf.~\eqref{ridg})
	\begin{equation}\label{fjupp}(R_{k,\tau}-R_{k,1})(z) = 2(1-\tau)\log|z|,\qquad z\in\D_e(\rho).\end{equation}

	Using the Jacobian \eqref{jac} and recalling that $\re h_k=\frac 1 2\log \Delta Q$ on $C_k$
	\begin{align}\int_{D_n^k}&(|\phi_k'|^{-1} e^{\re h_k})\circ \phi_k^{-1}\cdot e^{-nR_{k,\tau}}\, dA\nonumber\\
		&=\frac {1} \pi\int_{\tilde{D}_n}\frac {e^{\re h_k} \cdot e^{-n(R_{k,\tau}-R_{k,1})\circ\phi_k}}{\sqrt{2\Delta Q}}\circ \phi_k^{-1}\circ\psi_{k,t}(w)
		\cdot (1+\bigO(t))\, e^{-nt^2}\, dt\, |dw|\nonumber\\
		&=\frac {1} {\sqrt{2}\pi}\int_{\tilde{D}_n} (\frac {e^{\re h_k}} {\sqrt{\Delta Q}})\circ \phi_k^{-1} \circ \psi_{k,t} (w)\cdot |\psi_{k,t}(w)|^{2(n-j)}(1+\bigO(t))\, e^{-nt^2}\, dt\, |dw|.\label{redu}
	\end{align}
	
	Here and in the sequel we write ``$|dw|$'' or for the arclength measure on $\T$; for the curves $C_k$ we prefer to write ``$ds$''.

	Using that $|\psi_{k,t}(w)|=1+\bigO(t)$ uniformly on $\T$
	the expression \eqref{redu} is seen to be equal to
	$\sqrt{\tfrac{2\pi}{n}} \cdot (1+\bigO(\frac{n-j}{n}))$.

	In all, we have shown that
	\begin{align}\int_{\phi_1^{-1}(D_n^1)\cup \phi_2^{-1}(D_n^2)}&|\chi\Phi_{j,n}|^2e^{-n\tilde{Q}}\, dA\nonumber\\
		&=\sqrt{\frac{2\pi}{n}}( r_1^{2j+1}e^{-nq_1(\infty)} e^{-h_1(\infty)}+ e^sr_2^{2j+1}e^{-nq_2(\infty)}e^{-h_2(\infty)})\cdot (1+\bigO(\delta_n)).\label{mainc}
	\end{align}

	It remains to estimate the contribution coming from $\C\setminus (\phi_1^{-1}(D_n^1)\cup \phi_2^{-1}(D_n^2))$. To this end, the main work has been done in Lemma \ref{negl0}, which shows that the contribution
	coming from the ring-domains $\calA_{k,n}$ adds only a multiplicative error of the order $n^{-N}$ where $N>0$ is arbitrarily large:
	$$\int_{\calA_{k,n}}|\chi\Phi_{j,n}|^2e^{-n\tilde{Q}}\lesssim n^{-N-1/2}r_k^{2j+1}e^{-nq_k(\infty)},$$
	where the implied constant depends on the maximum of $1/\sqrt{\Delta Q}$ over $C_k$.
	
	This error term can be absorbed in the $\bigO$-term in \eqref{mainc} provided that $|s|$ is not too large; $|s|\le \log n$ will do.
	The remaining integral is easy to estimate, and leads to an error term which is exponentially small as $n\to\infty$. The details may be left to the reader.

	In summary, we have shown the following result.

	\begin{lem} \label{lemm} For $j = n +\bigO(\log^2 n)$ the squared norm in $L^2_{n\tilde{Q}}$ of $\chi \Phi_{j,n}$ satisfies
		\begin{align*}\| \chi \Phi_{j,n}\|^2_{n\tilde{Q}} =\sqrt{\frac{2\pi}{n}}&( r_1^{2j+1}e^{-nq_1(\infty)}e^{-h_1(\infty)}+ e^sr_2^{2j+1}e^{-nq_2(\infty)}e^{-h_2(\infty)})\cdot (1+\bigO(\delta_n)),
		\end{align*}
		where the $\bigO$-constant is uniform for $|s|\le \log n$ as $n\to \infty$.
	\end{lem}

	\subsection{Approximate orthogonality}
	We shall now show that the monic quasi-polynomial $\Phi_{j,n}$ (defined in \eqref{pjn}, \eqref{altpjn}) is approximately orthogonal to lower order polynomials.

	\begin{lem} \label{appo} Suppose that $j=n+\bigO(\log^2 n)$.
		If $p(z)$ is a polynomial of degree $\ell<j$ then
		$$|(p,\chi\Phi_{j,n})_{n\tilde{Q}}|\le C\delta_n\|p\|_{n\tilde{Q}}\|\chi \Phi_{j,n}\|_{n\tilde{Q}},$$
		where the constant $C$ is uniform for $|s|\le \log n$ and $\delta_n$ is as in \eqref{deltan}.
	\end{lem}

	\begin{proof} To ease the notation we first assume that $s=0$ and accordingly we write $Q$ instead of $\tilde{Q}$.

		For $k=1,2$ we let $P_k=P_{k,j,n,p}$ be the holomorphic function in $\D_e(\rho)$ such that $$\Lambda_{k,j,n}[P_k]=p.$$
		(See \eqref{posit}).

		Note that $P_k(z)=O(z^{\ell-j})$ as $z\to\infty$.

		The critical part in the estimation of the scalar product $(p,\chi\Phi_{j,n})_{n\tilde{Q}}$ comes from the ring-domains $\phi_k^{-1}(D_{n}^k)$ for $k=1,2$, i.e., from
		the two integrals
		$$I_k=\int_{\phi_k^{-1}(D_{n}^k)}p\overline{\Phi_{j,n}}\, e^{-nQ}\, dA,\qquad (k=1,2).$$

		To estimate the $I_k$'s, we use the Jacobian \eqref{jac} to write
		\begin{align*}I_k&=\int_{D_n^k}P_k\bar{f}_k e^{-R_{k,\tau}}\, dA\\
			&=\frac 1 \pi \int_{\tilde{D}_n}(\frac {P_k}{f_k})\circ \psi_{k,t}\cdot (|f_k|^2\cdot e^{-R_{k,\tau}})\circ \psi_{k,t}\cdot \frac {1}
			{\sqrt{2\Delta Q}}\circ \phi_k^{-1}\circ\psi_{k,t}\cdot (1+\bigO(t))\, dt|dw|.
		\end{align*}
		
		Inserting the expression \eqref{efk} for $f_k$, using \eqref{fjupp} and using $\frac{1}{2}\log \Delta Q=\re h_k$ on $C_k$, we rewrite the integral $I_k$ as
		$$r_k^{j+1/2}e^{-nq_k(\infty)/2}e^{-h_k(\infty)/2}\frac 1 {\pi\sqrt{2}}\int_{\tilde{D}_n}\frac {P_k}{\sqrt{\phi_k'\circ\phi_k^{-1}}}\circ \psi_{k,t}(w)\, e^{-nt^2}\cdot (1+\bigO(t))\, dt\,|dw|.
		$$

		But by the mean-value theorem, since $P_k(\infty)=0$
		$$\oint_\T \frac {P_k}{\sqrt{\phi_k'\circ\phi_k^{-1}}}\circ \psi_{k,t}(w)\,|dw|=
		\frac {P_k}{\sqrt{\phi_k'\circ\phi_k^{-1}}}\circ\psi_{k,t} (\infty)=0$$
		for $k=1,2$.
		
		It follows that
		$$|I_k|\lesssim r_k^{j+1/2}e^{-nq_k(\infty)/2} e^{-h_k(\infty)/2}\int_{\tilde{D}_n}\bigO(\delta_n)\cdot \Big|\frac {P_k}{\sqrt{\phi_k'\circ\phi_k^{-1}}}\circ\psi_{k,t}(w)\Big|\cdot e^{-nt^2}\, dt\, |dw|,$$
		where the implied constant depends on the maximum of $1/\sqrt{\Delta Q}$ over $C_k$.

		Since $1/\sqrt{\phi_k'}$ is bounded on $\phi_k^{-1}(\D_e(\rho))$ we obtain
		$$|I_k|\lesssim \delta_n r_k^{j+1/2}e^{-nq_k(\infty)/2}\int_{\tilde{D}_n}|P_k\circ\psi_{k,t}(w)| e^{-nt^2}\, dt\, |dw|.$$

		Using the Cauchy-Schwarz inequality we now deduce that
		\begin{align*}|I_k|&\lesssim \delta_n r_k^{j+1/2}e^{-nq_k(\infty)/2}\Bigl(\int_\R e^{-nt^2}\, dt\Bigr)^{1/2}\Bigl(\int_{D_{n}^k}|P_k|^2e^{-nR_\tau}\, dA\Bigr)^{1/2}\\
			&\lesssim \frac {\sqrt{\log n}}{n^{3/4}}r_k^{j+1/2}e^{-nq_k(\infty)/2}\|p\|_{nQ}.
		\end{align*}

		In view of Lemma \ref{lemm} it follows that
		$$|I_1+I_2|\lesssim \delta_n\|p\|_{nQ}\|\chi\Phi_{j,n}\|_{nQ}.$$
		We stress that a similar estimate is true for any $s$, by replacing $Q$ by $\tilde{Q}$.

		To estimate the contribution coming from integration over $\C\setminus(\phi_1^{-1}(D_n^1)\cup \phi_2^{-1}(D_n^2))$, we argue similarly as for the norm-computations as in the previous section. Using Lemma \ref{negl0},
		and the Cauchy-Schwarz inequality, we deduce that if $|s|\le \log n$
		\begin{align*}\Big|\int_{\C\setminus (\phi_1^{-1}(D_{n}^1)\cup \phi_2^{-1}(D_{n}^2))}p\overline{\chi\Phi_{j,n}}\, e^{-n\tilde{Q}}\, dA\,\Big|\lesssim n^{-N}\|p\|_{n\tilde{Q}}\|\chi\Phi_{j,n}\|_{n\tilde{Q}},\end{align*}
		where $N>0$ is as large as we please. Adding up the contributions, we finish the proof of the lemma.
	\end{proof}

	\subsection{Correcting $\Phi_{j,n}$ to a monic polynomial} For $j$ close to $n$, we shall now estimate the norm $\tilde{h}_{j,n}(s)$.
	We accomplish this by approximating with the quasi-polynomial $\chi\Phi_{j,n}$ using
	an estimate for the solution to a constrained $\dbar$-problem. We turn to the details.

	Let $u$ be the norm-minimal solution in $L^2_{n\tilde{Q}}$ to the $\dbar$-problem:
	\begin{enumerate}
		\item \label{1} $\dbar u=\dbar \chi\cdot \Phi_{j,n}$,
		\item \label{2} $u(z)=\bigO(z^{j-1})$ as $z\to\infty$.
	\end{enumerate}

	\begin{lem} \label{bresk} For $|s|\le \log n$ and $j=n+\bigO(\log^2 n)$ we have the estimate
		\begin{equation}\label{fagel}\|u\|_{n\tilde{Q}}^2\lesssim n^{\ell}\int_\C |\dbar\chi|^2\cdot |\Phi_{j,n}|^2e^{-n\tilde{Q}}\, dA,\end{equation}
		where $\ell$ depends only on the sup-norm $\|\omega\|_\infty$ and the implied constant depends only on $Q$.
	\end{lem}

	\begin{proof} We apply the technique of H\"{o}rmander estimates (e.g.~\cite{HW}) but with a slight extra twist.

		First apply the estimate in \cite[Section IV.4.2]{Hor} to the strictly subharmonic weight function
		$$\phi(z)=n\check{Q}_\tau(z)+\eps\log(1+|z|^2),$$
		where we take $\tau=j/n$ and $0<\eps<1/2$.

		For large $n$, $\dbar\chi$ is supported in the $\tau$-droplet $S_\tau=\{\check{Q}_\tau=Q\}$.
		
		It follows
		from \cite[(4.2.6)]{Hor} that the $L^2_\phi$-minimal solution $u_0$ to the problem $\dbar u_0=\dbar(\chi\Phi_{j,n})$ satisfies
		\begin{equation}\label{phino}\int_\C|u_0|^2e^{-\phi}\, dA\lesssim \int_\C\frac {|\dbar\chi|^2|\Phi_{j,n}|^2}{\Delta \phi}e^{-\phi}\, dA.\end{equation}

		Using the growth of $(Q-\check{Q}_\tau)(z)$ near infinity, it follows that
		\begin{equation}\label{gruk0}\int_\C|u_0|^2e^{-nQ}\, dA\lesssim \frac 1 n \int_\C|\dbar\chi|^2|\Phi_{j,n}|^2e^{-nQ}\, dA.\end{equation}

		Next, since the function $\wfun$ is bounded, there is a constant $C$ such that, as $n\to\infty$
		$$n^{-C}\lesssim e^{s\wfun}\lesssim n^C,\qquad (|s|\le\log n).$$
		It follows from this and \eqref{gruk0} that
		\begin{equation}\label{gruk1}\int_\C|u_0|^2e^{-n\tilde{Q}}\, dA\lesssim  n^\ell\int_\C|\dbar\chi|^2|\Phi_{j,n}|^2e^{-n\tilde{Q}}\, dA,\end{equation}
		where $\ell$ is an absolute constant.

		Now note that the function $\Psi_{j,n}:=u_0-\dbar(\chi\Phi_{j,n})$ is entire. Moreover, by our choice of $\tau$ and \eqref{phino}, we have that $$\int_\C|\Psi_{j,n}(z)|^2\frac {dA(z)}{(1+|z|^2)^{j+\eps}}<\infty.$$
		This implies that $\Psi_{j,n}(z)/z^{j-1}$ has a removable singularity at infinity, i.e., $\Psi_{j,n}(z)$ must be a polynomial of degree at most $j-1$.

		It follows that
		$u_0$ satisfies both \eqref{1} and \eqref{2} above.
	\end{proof}

	Since $j=n+\bigO(\log^2 n)$, we know from Lemma \ref{lemm} that the main contribution to the norm of $\chi\Phi_{j,n}$ comes from integration over $\phi_1^{-1}(D_n^1)\cup\phi_2^{-1}(D_n^2)$.

	Moreover, for large $n$ we have (by choosing $\delta>0$ somewhat smaller if necessary) that the support of $\dbar\chi$ is contained in the set $\calA_{1,n}$ in Lemma \ref{negl0}.
	It hence follows easily from Lemma \ref{bresk} and Lemma \ref{negl0} that
	\begin{equation}\label{bla}\|u\|_{n\tilde{Q}}^2\lesssim n^{-N_1+\ell}\|\chi\Phi_{j,n}\|_{n\tilde{Q}}^2,\end{equation}
	$N_1>0$ is arbitrary and the implied constant is uniform for $|s|\le\log n$.

	We next define an entire function $E_{j,n}$ by
	$$E_{j,n}(z)=\chi(z)\Phi_{j,n}(z)-u(z),$$
	and observe that $E_{j,n}(z)=z^j+O(z^{j-1})$ as $z\to\infty$. Hence $E_{j,n}(z)$ is a monic polynomial of degree $j$,
	and (with $N=N_1-\ell$)
	$$\|E_{j,n}-\chi \Phi_{j,n}\|_{n\tilde{Q}}\lesssim n^{-N}\|\chi\Phi_{j,n}\|_{n\tilde{Q}},\qquad (n\to\infty).$$

	It follows that
	$$\|E_{j,n}\|_{n\tilde{Q}}=\|\chi\Phi_{j,n}\|_{n\tilde{Q}}\cdot (1+\bigO(n^{-N})).$$

	Moreover, by Lemma \ref{appo}, if $p(z)$ is a polynomial of degree at most $j-1$, then
	\begin{align}|(p,E_{j,n})_{n\tilde{Q}}|&\le|(p,\chi \Phi_{j,n})_{n\tilde{Q}}|+|(p,E_{j,n}-\chi \Phi_{j,n})_{n\tilde{Q}}|\nonumber \\
		&\lesssim \delta_n\|p\|_{n\tilde{Q}}\|E_{j,n}\|_{n\tilde{Q}},\label{malign}
	\end{align}
	where the implied constant is uniform for $|s|\le \log n$.

	Now let $\tilde{\pi}_{j-1,n}$ be the orthogonal projection of $L^2_{n\tilde{Q}}$ on the subspace of polynomials of degree at most $j-1$. We get that
	$$E_{j,n}-\tilde{\pi}_{j-1,n}(E_{j,n}) = \tilde{p}_{j,n},$$
	 and also by \eqref{malign} we have $\|\tilde{\pi}_{j-1,n}(E_{j,n})\|_{n\tilde{Q}}\lesssim\delta_n\|E_{j,n}\|_{n\tilde{Q}}$, so
	$$\|\tilde{p}_{j,n}-E_{j,n}\|_{n\tilde{Q}}\lesssim \delta_n\|E_{j,n}\|_{n\tilde{Q}}.$$

	Altogether, we have found that
	\begin{equation}\label{bpo}\|\tilde{p}_{j,n}-\chi \Phi_{j,n}\|_{n\tilde{Q}}\lesssim \delta_n\|\chi \Phi_{j,n}\|_{n\tilde{Q}}.\end{equation}
	As a consequence,
	$$\|\tilde{p}_{j,n}\|_{n\tilde{Q}}=\|\chi\Phi_{j,n}\|_{n\tilde{Q}}\cdot (1+\bigO(\delta_n)).$$

	We obtain the following theorem.

	\begin{thm} \label{insert}  If $j=n+\bigO(\log^2 n)$ then
			$$\tilde{h}_{j,n}(s)=\sqrt{\frac{2\pi}{n}}( r_1^{2j+1}e^{-nq_1(\infty)}e^{-h_1(\infty)}+ e^s r_2^{2j+1}e^{-nq_2(\infty)}e^{-h_2(\infty)})\cdot (1+\bigO(\delta_n)),$$
		where the implied constant is uniform for $|s|\le \log n$.
	\end{thm}
	
	\subsection{Asymptotics for wavefunctions in the bifurcation regime}
	From \eqref{bpo} we can also obtain pointwise estimates for the monic orthogonal polynomial $p_{j,n}(z)$.
	
	However, results become more transparent if we pass to the $j$th weighted orthonormal polynomial (or \textit{wavefunction})
	\begin{equation}\label{wavefunction}w_{j,n}(z)=w_{j,n}(z;s):=\gamma_{j,n}\tilde{p}_{j,n}(z)e^{-n\tilde{Q}(z)/2}\end{equation}
	where the leading coefficient $\gamma_{j,n}>0$ is chosen so that $\{w_{j,n}\}_{j=0}^{n-1}$ is an orthonormal basis for the subspace
	$$\calW_n:=\{p(z)e^{-n\tilde{Q}(z)/2}\,;\, p\,\text{is a polynomial of degree at most}\, n-1\}$$ of $L^2(\C, dA)$.
	
	We will here focus on asymptotics in the bifurcation regime.
	
	In this case, all interesting asymptotics takes place in the domain $\calB_1\cup \Ext C_1$ where $\calB_1=\{z\,;\,\dist(z,C_1)\le M\delta_n\}$ for a suitable (large) constant $M$.
	
	Intuitively, $|w_{j,n}(z)|$ peaks along two curves, one near $C_1$ and another near $C_2$, and is negligible outside a small neighbourhood of $C_1\cup C_2$, and the maximum size is of order $\bigO(n^{1/4})$.
	
	To describe the situation more concretely, we first note that, by Theorem \ref{insert}, the leading coefficient $\gamma_{j,n}$ satisfies
	$\gamma_{j,n}=c_{j,n}^{-1/2}\cdot (1+\bigO(\delta_n))$, where
	\begin{align}c_{j,n}:=\sqrt{\frac{2\pi}{n}}( r_1^{2j+1}e^{-nq_1(\infty)}e^{-h_1(\infty)}+ e^s r_2^{2j+1}e^{-nq_2(\infty)}e^{-h_2(\infty)}).\end{align}
	
	It is therefore natural to approximate $w_{j,n}(z)$ by the function
	\begin{equation}\label{wwf}F_{j,n}(z):=c_{j,n}^{-1/2}\Phi_{j,n}(z)e^{-n\tilde{Q}(z)/2}.\end{equation}

	We have the following theorem.

	\begin{thm} \label{gondor} Suppose that $j=n+\bigO(\log^2 n)$. 
		Then for
		$z\in \calB_1\cup\Ext C_1$
		\begin{equation}\label{fjodor}|w_{j,n}(z)-F_{j,n}(z)|\lesssim (\sqrt{\log n})
			e^{-n(Q-\check{Q}_{\tau})(z)/2}.\end{equation}
	\end{thm}
	
	\begin{proof} First assume that $z\in\calB_1$.
		From a standard pointwise-$L^2$ estimate (see for example \cite[Lemma 2.4]{A})
		$$|\tilde{p}_{j,n}(z)-\Phi_{j,n}(z)|e^{-n\tilde{Q}(z)/2}\lesssim \sqrt{n} \|\tilde{p}_{j,n}-\chi\Phi_{j,n}\|_{n\tilde{Q}}.$$
		From \eqref{bpo} and Theorem \ref{insert} we thus have
		\begin{equation}\label{bby}|p_{j,n}(z)-\Phi_{j,n}(z)|e^{-n\tilde{Q}(z)/2}\lesssim \sqrt{\log n}\sqrt{c_{j,n}}
			.\end{equation}
		
		Now consider the outer boundary of the $\tau$-droplet,
		$$\Gamma_{\tau}=\d_* S_{\tau}.$$
		For $n$ large enough we have $\Gamma_\tau\subset \calB_1$.

		For $z\in\Ext \Gamma_{\tau}$ we introduce the analytic function
		$$q(z)=\tilde{p}_{j,n}(z)-\Phi_{j,n}(z)$$ and the weighted counterpart
		$f=q\cdot e^{-n\tilde{Q}/2}$. Observe that
		$$\frac 1 n \log|q(z)|^2=\frac 1 n\log |f(z)|^2+\tilde{Q}(z),$$
		where the left hand side $u(z):=\frac 1 n\log|q(z)|^2$ is subharmonic on $\Ext(\Gamma_{\tau})$ and $u(z)\lesssim \tau \log|z|^2$ as $z\to\infty$.
		
		On $\Gamma_{\tau}$ we have $u\le (\const+K_n )/n+Q$ where $K_n=\log(c_{j,n}\log n)$.
		
		A suitable version of the maximum principle implies that $u\le (\const+K_n)/n+\hat{Q}_{\tau}$ on $\Ext(\Gamma_\tau)$, so that \eqref{bby} holds in this case as well.
		
		Dividing through by $\sqrt{c_{j,n}}$ in \eqref{bby} and recalling that $\gamma_{j,n}=c_{j,n}^{-1/2}\cdot (1+\bigO(\delta_n))$, we finish the proof.
	\end{proof}
	
	\begin{rem} From the form of $\Phi_{j,n}(z)$ (see \eqref{pjn}, \eqref{altpjn}) it is not hard to see that the maximum size of $w_{j,n}(z)$ is $\bigO(n^{1/4})$.
		Since $\sqrt{\log n}\ll n^{1/4}$, $F_{j,n}(z)$ gives a good approximation near the peak-set of $w_{j,n}(z)$.
	\end{rem}

		\subsection{Asymptotics for lower degree orthogonal polynomials} \label{alo} We shall now consider lower order terms $\tilde{h}_{j,n}(s)=\|\tilde{p}_{j,n}\|_{n\tilde{Q}}^2$ for $j\le n-\log^2 n$. Our goal is to prove that $\log(\tilde{h}_{j,n}(s)/\tilde{h}_{j,n}(0))$ is ``negligible'' as $n\to\infty$.

	For given $\tau\le 1$ we write $\check{Q}_\tau$ for the obstacle function in potential $Q$ which increases near infinity as
	$$\check{Q}_\tau(z)=2\tau\log|z|+\bigO(1), \qquad (z\to\infty).$$
	We refer to \cite[Section 3]{AC} for details about $\check{Q}_\tau$.

	We recall the following pointwise-$L^2$ estimate from \cite[Lemma 3.7]{AC}.

	\begin{lem}\label{basic} Let $p(z)$ be a polynomial of degree $j\le n$ and let $\tau =  \frac{j}{n}$. Then there is a constant $C$ (depending on $Q$) such that for all $z\in\C$
		$$|p(z)|e^{-nQ(z)/2}\le C\sqrt{n}\|p\|_{nQ}e^{-n(Q-\check{Q}_\tau)(z)/2}.$$
	\end{lem}

		Let $\mathcal{E}$ be a small but fixed neighbourhood of the droplet $S$ on which $\omega =0$.
		
		\begin{lem}\label{small_out} If $p(z)$ is a polynomial of degree $j\le n- \log^2 n$ and $|s|\le \log n$ then
			$$\int_{\C\setminus \calE}|p|^2e^{-n\tilde{Q}}\, dA\lesssim n^{-N} \|p\|_{n\tilde{Q}}^2 \,,$$
			where $N>0$ is arbitrary and the implied constant depends only on $N$.
		\end{lem}
		
		\begin{proof}
			It follows from \cite[Lemma 4.4]{AC} that
			\begin{equation}\label{out_est}
				e^{n(\check{Q}_\tau(z)-\check{Q}(z))} = \big(1 + \bigO(n(1-\tau)^2)\big) \cdot |\phi_1(z)|^{2(j-n)},
			\end{equation}
			as $n\to \infty$, locally uniformly on $\C\setminus \calE$.
			
			Using Lemma \ref{basic} we obtain
			$$
			\int_{\C\setminus \calE}|p|^2e^{-n\tilde{Q}}\, dA \lesssim \sqrt{n} e^{s} \|{p\|}_{n\tilde{Q}}^2 \int\limits_{\C\setminus \calE} e^{-n(Q(z)-\check{Q}_\tau(z))} dA.
			$$
			
			On a neighbourhood $\mathcal{N}$ of the outpost we use that  $\check{Q}(z)-Q(z) \leq 0$ and \eqref{out_est} to obtain
			$$
			\int\limits_{\mathcal{N}} e^{-n(Q(z)-\check{Q}_\tau(z))} dA \lesssim \int\limits_{\mathcal{N}} |\phi_1(z)|^{2(j-n)} dA = \bigO(n^{-N}),
			$$
			as $n\to \infty$. In the last step we have used that $|\phi_1(z)|\geq c>1$ on $\calN$.
			
			For $z\in(\C\setminus\calE)\setminus \mathcal{N}$ we have $\check{Q}_\tau(z)-Q(z) < -c$ for some constant $c>0$. Together with the growth assumption on $Q$ in \eqref{growth_Q} this gives the same estimate in $(\C\setminus\calE)\setminus \mathcal{N}$.

		\end{proof}

		Let us write $p_{j,n}(z)$ and $\tilde{p}_{j,n}(z)$ for the monic orthogonal polynomials of degree $j$ in $L^2_{nQ}$ and $L^2_{n\tilde{Q}}$ respectively, where $\tilde{Q}=Q-\frac s n \omega$, $|s|\le \log n$.
		
		Recall that
		$\tilde{h}_{j,n}(0)=\|p_{j,n}\|_{nQ}^2$ and $\tilde{h}_{j,n}(s)=\|\tilde{p}_{j,n}\|_{n\tilde{Q}}^2$.

		\begin{lem} \label{L-lemma} Suppose that $j\le n-\log^2 n$. Then
			$$\log \frac {\tilde{h}_{j,n}(s)}{\tilde{h}_{j,n}(0)}=\bigO(n^{-N}),\qquad (n\to\infty)$$
			where $N>0$ is arbitrarily large and the implied constant depends only on $N$.
		\end{lem}

		\begin{proof} With $j$ satisfying the assumption, let $q(z)$ be a polynomial of degree less than $j$.
			
			Then $(p_{j,n},q)_{nQ}=0$. But since $\omega=0$ on $\calE$, Lemma \ref{small_out} and the Cauchy-Schwarz inequality imply
			$$0=(p_{j,n},q)_{nQ}=(p_{j,n},q)_{n\tilde{Q}}+\|p_{j,n}\|_{nQ}\|q\|_{nQ}\cdot \bigO(n^{-N}).$$
			
			Now consider the orthogonal projection $\tilde{\pi}_{j-1,n}$ of $L^2_{n\tilde{Q}}$ onto the space of polynomials of degree at most $j-1$. The above shows that
			$$\|\tilde{\pi}_{j-1,n}(p_{j,n})\|=\bigO(n^{-N}).$$
			
			Since both $p_{j,n}$ and $\tilde{p}_{j,n}$ are monic polynomials, we have
			$$p_{j,n}-\tilde{\pi}_{j-1,n}(p_{j,n})=\tilde{p}_{j,n}.$$
			By integration, using again that $\omega=0$ on $\calE$, we find that
			$$\tilde{h}_{j,n}(0)+\tilde{h}_{j,n}(0)\cdot \calO(n^{-N})=\tilde{h}_{j,n}(s),$$
			which finishes the proof.
	\end{proof}

	\section{Number of points near the outpost} \label{proof: outpost} We are now ready to prove Theorem \ref{main1}.

	Let us write
	$$h_{j,n}^0=r_1^{2j+1}e^{-nq_1(\infty)}e^{-h_1(\infty)},\qquad h_{j,n}^1=r_2^{2j+1}e^{-nq_2(\infty)}e^{-h_2(\infty)}.$$

	Inserting the asymptotics in Theorem \ref{insert} (using Lemma \ref{L-lemma} to discard lower-order terms) we deduce that the cumulant generating function $F_{n,\wfun}$ of $N_n=\fluct_n \wfun$ satisfies (with $\delta_n$ given by \eqref{deltan})
	\begin{align*}F_{n,\wfun}(s)&=\log Z_{n,s\wfun}-\log Z_{n,0}\\
		&=\sum_{j=n-\log^2 n}^{n-1}\log\biggl(1+e^s\frac {h_{j,n}^1}{h_{j,n}^0}\biggr)-\sum_{j=n-\log^2n}^{n-1}\log\biggl(1+\frac {h_{j,n}^1}{h_{j,n}^0}\biggr)+\bigO(\delta_n\log^2n).
	\end{align*}

	However since $h_2(\infty)-h_1(\infty)=c$
	$$\sum_{j=n-\log^2 n}^{n-1}\log\biggl(1+e^s\frac {h_{j,n}^1}{h_{j,n}^0}\biggr)=\sum_{j=n-\log^2 n}^{n-1}\log \biggl(1+e^se^{-c}\bigl(\frac {r_2}{r_1}\bigr)^{2j+1}e^{n(q_1(\infty)-q_2(\infty))}\biggr).$$

	From Lemma \ref{lura}, $$\frac {e^{nq_1(\infty)}}{e^{nq_2(\infty)}}=\bigl(\frac {r_1} {r_2}\bigr)^{2n}$$ so the last sum above simplifies as
	$$\sum_{j=n-\log^2 n}^{n-1}\log\biggl(1+e^{s}e^{-c}\bigl(\frac {r_2}{r_1}\bigr)^{2(j-n)+1}\biggr).$$

	Shifting the summation index, the last sum is recognized in terms of the $q$-Pochhammer symbol \eqref{qpoc}:
	$$\log\biggl[\bigl(-e^{s}e^{-c}\frac{r_1}{r_2};\bigl(\frac {r_1}{r_2}\bigr)^2\bigr)_\infty \biggr]+\bigO(n^{-N}),$$
	where the integer $N$ may be taken as large as we please.

	Comparing with the formula \eqref{cgf_he} for the cumulant generating function of a Heine distributed random variable $X$ with parameters
	$$\theta=\frac {r_1}{r_2} e^{-c},\qquad q=\bigl(\frac {r_1}{r_2}\bigr)^2,$$
	we now recognize that, as $n\to\infty$,
	$F_{n,\wfun}(s)=F_X(s)+\bigO(\delta_n\log^2n ).$

	The proof of Theorem \ref{main1} is complete. q.e.d.

	\section{Bifurcation for orthogonal polynomials near a spectral gap}\label{bifur2}
	In the remainder of this note, we assume that $Q(z)$ is a spectral gap potential as in Section \ref{csg}. We also fix $s>0$ and consider the perturbed potential
	$$\tilde{Q}=Q-\frac s n \wfun$$
	where $\wfun=0$ in a neighbourhood of $S_{\tau_*}$ and $\wfun=1$ in a neighbourhood of $S\setminus S_{\tau_*}$.

	The monic orthogonal polynomial $p_{j,n}(z)$ of degree $j$ in $L^2_{n\tilde{Q}}$ undergoes a bifurcation
	as $j$ passes the critical value $n\tau_*$. As the details are similar to the case of an outpost potential,
	we will allow ourselves to be brief.
	
	For the critical regime of $j$ with
	\begin{equation}\label{voff}|j-n\tau_*|\le \log^2 n\end{equation} we consider monic quasi-polynomials of the form
	$$\Phi_{j,n}(z)=\frac {r_{1,\tau_*}^{j+1/2}}{e^{nq_{1,\tau_*}(\infty)} e^{h_1(\infty)/2}}\sqrt{\phi_{1,\tau_*}'(z)}\phi_{1,\tau_*}(z)^je^{nq_{1,\tau_*}(z)/2}e^{h_1(z)/2}.$$
	
	Here, for $k=1,2$, the function $q_{k,\tau_*}(z)$ is the bounded holomorphic function on $\Ext C_{k}$ which is real at infinity and solves the Dirichlet problem
	$$\re q_{k,\tau_*}=Q\qquad \text{on}\qquad C_{k}.$$

	Note that if $\check{Q}_{\tau_*}$ is the obstacle function which grows like $2\tau_*\log|z|+\bigO(1)$ as $z\to\infty$ then (for $z\in\Ext C_2$)
	$$\check{Q}_{\tau_*}(z)=2\tau_*\log|\phi_{1,\tau_*}(z)|+\re q_{1,\tau_*}(z)=2\tau_*\log|\phi_{2,\tau_*}(z)|+\re q_{2,\tau_*}(z).$$

	By arguing as in Lemma \ref{lura}, we now obtain the following result.

	\begin{lem} \label{lura2} For $z\in \Ext C_{2}$ we have the identity (Lemma \ref{lura2})
		$$\phi_{1,\tau_*}(z)e^{\frac {\tau_*} 2 q_{1,\tau_*}(z)}=\phi_{2,\tau_*}(z)e^{\frac {\tau_*} 2 q_{2,\tau_*}(z)},$$
		and
		$$\frac 1 {r_{1,\tau_*}}e^{\frac {\tau_*} 2 q_{1,\tau_*(\infty)}}=\frac 1 {r_{2,\tau_*}}e^{\frac {\tau_*} 2 q_{2,\tau_*(\infty)}}.$$
		For $z\in \Ext C_{2}$ we have the alternative formula
		$$\Phi_{j,n}(z)=\frac {r_{2,\tau_*}^{j+1/2}}{e^{nq_{2,\tau_*}(\infty)}e^{h_2(\infty)/2}}\sqrt{\phi_{2,\tau_*}'(z)}\phi_{2,\tau_*}(z)^je^{nq_{2,\tau_*}(z)/2}e^{ h_2(z)/2}.$$

	\end{lem}

	Arguing as Section \ref{bifur1}, it is not hard to prove that if $j$ is in the regime \eqref{voff}, then $\chi\Phi_{j,n}$
	is a good approximation to the monic orthogonal polynomial $p_{j,n}$, where the smooth function $\chi$ is identically $1$ in a small neighbourhood of $\Ext (C_1)$ and vanishes outside a slightly larger neighbourhood.
	
	If $j$ is outside of this bifurcation regime, we can instead rely on a small modification of Lemma \ref{L-lemma}.

	We now state our main result on asymptotics for the squared norm of the order $j$ monic orthogonal polynomial
	$$\tilde{h}_{j,n}(s):=\|\tilde{p}_{j,n}\|_{n\tilde{Q}}^2.$$

	\begin{thm} \label{insert2} Assume that  $|j-n\tau_*|\le \log^2 n$, then as $n\to \infty$, the norm $\tilde{h}_{j,n}(s)$ satisfies the following asymptotic
			$$\tilde{h}_{j,n}(s)=\sqrt{\frac{2\pi}{n}}( r_{1,\tau_*}^{2j+1}e^{-nq_{1,\tau_*}(\infty)}e^{-h_1(\infty)}+ e^s r_{2,\tau_*}^{2j+1}e^{-nq_{2,\tau_*}(\infty)}e^{-h_2(\infty)})\cdot (1+\bigO(\delta_n)),$$
		where the implied constants are uniform for $|s|\le \log n$. Here $\delta_n=\sqrt{\log n/n}$.
	\end{thm}

	It is straightforward to write down a detailed proof of Theorem \ref{insert2} which resembles Theorem \ref{insert}; we omit details here, to avoid tedious repetitions.
	
	\smallskip
	
	We also note the following theorem on asymptotics for the wavefunctions $w_{j,n}(z)=\gamma_{j,n}p_{j,n}(z)e^{-n\tilde{Q}(z)/2}$ in the bifurcation regime. (The proof of Theorem \ref{gondor} goes through essentially unchanged.)
	
	\begin{thm}\label{gondor2} Suppose that $j=n\tau_*+\bigO(\log^2n)$ and let $z\in \calB_1\cup\Ext C_1$.
		Then, with the natural interpretations (substituting $r_k\mapsto r_{k,\tau_*}$ and so on) the asymptotics in Theorem \ref{gondor} holds for $w_{j,n}(z)$.
	\end{thm}

	\section{Number of particles near the ring $S\setminus S_{\tau_*}$} \label{numgap}
	We now prove Theorem \ref{mth2}. To this end, our main tool is Theorem \ref{insert2}.

	Let us write
	$$h_{j,n}^0:=r_{1,\tau_*}^{2j+1}e^{-nq_{1,\tau_*}(\infty)}e^{-h_1(\infty)},\qquad h_{j,n}^1:=r_{2,\tau_*}^{2j+1}e^{-nq_{2,\tau_*}(\infty)}e^{-h_2(\infty)}.$$

	By the formula \eqref{cgf_ring} and Theorem \ref{insert2}, the cumulant generating function of $\fluct_n\wfun$ satisfies the asymptotic
	\begin{align*}&F_{n,\wfun}(s)=\log Z_{n,s\wfun}-ns(1-\tau_*)-\log Z_{n,0}\\
		&=\sum_{|j-n\tau_*|\le \log^2 n }\log(h_{j,n}^0+e^sh_{j,n}^1)-ns(1-\tau_*)-\sum_{|j-n\tau_*|\le \log^2 n }\log(h_{j,n}^0+h_{j,n}^1)+\bigO(\delta_n\log^2 n )\nonumber \\
		&=\sum_{j=\lfloor n\tau_*- \log^2 n\rfloor}^{\lfloor n\tau_*\rfloor-1}\log(h_{j,n}^0+e^sh_{j,n}^1)-\sum_{j=\lfloor n\tau_*-\log^2 n \rfloor}^{\lfloor n\tau_*\rfloor-1}\log(h_{j,n}^0+h_{j,n}^1)\nonumber \\
		&+\sum_{j=\lfloor  n\tau_*\rfloor }^{\lfloor n\tau_*+\log^2 n \rfloor}\log(e^{-s}h_{j,n}^0+h_{j,n}^1)-\sum_{j=\lfloor n\tau_*\rfloor}^{\lfloor n\tau_*+\log^2 n \rfloor}\log(h_{j,n}^0+h_{j,n}^1)+\bigO(\delta_n\log^2 n ).\nonumber
	\end{align*}

	To relate the above to the Heine distribution, we first use the identity
	$$\frac {e^{nq_{1,\tau_*}(\infty)}}{e^{nq_{2,\tau_*}(\infty)}}=\bigl(\frac {r_{1,\tau_*}}{r_{2,\tau_*}}\bigr)^{2n\tau_*},$$
	to write
	\begin{align*}
		&\sum_{j=\lfloor n\tau_*-\log^2 n \rfloor}^{\lfloor n\tau_*\rfloor-1}\log(h_{j,n}^0+e^sh_{j,n}^1)-\sum_{j=\lfloor n\tau_*-\log^2 n \rfloor}^{\lfloor n\tau_*\rfloor-1}\log(h_{j,n}^0+h_{j,n}^1)\\
		&=\sum_{j=\lfloor n\tau_*-\log^2 n \rfloor}^{\lfloor n\tau_*\rfloor-1}\log\bigl(1+e^s\frac {h_{j,n}^1}{h_{j,n}^0}\bigr)-\sum_{j=\lfloor n\tau_*-\log^2 n \rfloor}^{\lfloor n\tau_*\rfloor-1}\log\bigl(1+\frac {h_{j,n}^1}{h_{j,n}^0}\bigr)\\
		&=\sum_{j=\lfloor n\tau_*-\log^2 n \rfloor}^{\lfloor n\tau_*\rfloor-1}\log\bigl(1+e^se^{-c}\bigl(\frac {r_{2,\tau_*}}{r_{1,\tau_*}}\bigr)^{2(j-n\tau_*)+1}\bigr)\\
		&\qquad \qquad -\sum_{j=\lfloor n\tau_*-\log^2 n \rfloor}^{\lfloor n\tau_*\rfloor-1}\log\bigl(1+e^{-c}\bigl(\frac {r_{2,\tau_*}}{r_{1,\tau_*}}\bigr)^{2(j-n\tau_*)+1}\bigr).
	\end{align*}

	Shifting the summation index and writing
	$$x_n=\{n\tau_*\}=n\tau_*-\lfloor n\tau_*\rfloor$$
	we rewrite the above as
	\begin{align*}\sum_{j=0}^{\lfloor \log^2 n \rfloor}&\log \bigl(1+e^se^{-c}\bigl(\frac {r_{1,\tau_*}}{r_{2,\tau_*}}\bigr)^{2j+1+2x_n}\bigr)\\
		&\quad\quad -\sum_{j=0}^{\lfloor \log^2 n \rfloor}\log \bigl(1+e^{-c}\bigl(\frac {r_{1,\tau_*}}{r_{2,\tau_*}}\bigr)^{2j+1+2x_n}\bigr)+\bigO(n^{-N})\\
		&=\log\bigl[\bigl(-e^se^{-c}\bigl(\frac {r_{1,\tau_*}}{r_{2,\tau_*}}\bigr)^{1+2x_n};\bigl(\frac {r_{1,\tau_*}}{r_{2,\tau_*}}\bigr)^2\bigr)_\infty\bigr]\\
		&\qquad\qquad-\log\bigl[\bigl(-e^se^{-c}\bigl(\frac {r_{1,\tau_*}}{r_{2,\tau_*}}\bigr)^{1+2x_n};\bigl(\frac {r_{1,\tau_*}}{r_{2,\tau_*}}\bigr)^2\bigr)_\infty\bigr]+\bigO(n^{-N}),\end{align*}
	where $N>0$ is at our disposal. (As always $(z;q)_\infty$ denotes the $q$-Pochhammer symbol \eqref{qpoc}.)

	By Lemma \ref{vwe}, we recognize the above as an approximation (to within $\bigO(n^{-N})$)
	of the cumulant generating function of $X^+$ where $X^+_n\in \He(\theta_n^+,q)$,
	$$\theta_n^+=e^{-c}\bigl(\frac {r_{1,\tau_*}}{r_{2,\tau_*}}\bigr)^{1+2x_n},\qquad q=\bigl(\frac {r_{1,\tau_*}}{r_{2,\tau_*}}\bigr)^2.$$

	Similarly,
	\begin{align*}
		&\sum_{j=\lfloor n\tau_*\rfloor }^{\lfloor n\tau_*+\log^2 n \rfloor}\log(e^{-s}h_{j,n}^0+h_{j,n}^1)-\sum_{j=\lfloor n\tau_*\rfloor}^{\lfloor n\tau_*+\log^2 n \rfloor}\log(h_{j,n}^0+h_{j,n}^1)\\
		&=\sum_{j=\lfloor n\tau^*\rfloor}^{\lfloor n\tau^*+\log^2 n \rfloor}\log\bigl(e^{-s}\frac {h_{j,n}^0}{h_{j,n}^1}+1\bigr)-\sum_{j=\lfloor n\tau^*\rfloor}^{\lfloor n\tau^*+\log^2 n\rfloor}\log\bigl(\frac {h_{j,n}^0}{h_{j,n}^1}+1\bigr)\\
		&=\sum_{j=\lfloor n\tau^*\rfloor}^{\lfloor n\tau^*+\log^2 n \rfloor}\log \bigl(1+e^{-s}e^{c}\bigl(\frac {r_{1,\tau_*}}{r_{2,\tau_*}}\bigr)^{2(j-n\tau_*)+1}\bigr)\\
		&\qquad\qquad-\sum_{j=\lfloor n\tau^*\rfloor}^{\lfloor n\tau^*+\log^2 n \rfloor}\log\bigl(1+e^c\bigl(\frac {r_{1,\tau_*}}{r_{2,\tau_*}}\bigr)^{2(j-n\tau_*)+1}\bigr).
	\end{align*}

	Shifting the summation index the above becomes
	\begin{align*}\sum_{j=0}^{\lfloor\log^2 n \rfloor}\log\bigl(1+e^{-s}e^c\bigl(\frac {r_{1,\tau_*}}{r_{2,\tau_*}}\bigr)^{2j+1-2x_n}\bigr)-\sum_{j=0}^{\lfloor\log^2 n \rfloor}\log\bigl(1+e^c\bigl(\frac {r_{1,\tau_*}}{r_{2,\tau_*}}\bigr)^{2j+1-2x_n}\bigr).\end{align*}

	By Lemma \ref{vwe}, we recognize the above, to within an error of $\bigO(n^{-N})$, as
	the cumulant generating function of $X^-_n\sim \He(\theta_n^-,q)$ where
	$$\theta_n^-=e^c\bigl(\frac {r_{1,\tau_*}}{r_{2,\tau_*}}\bigr)^{1-2x_n},\qquad q=\bigl(\frac {r_{1,\tau_*}}{r_{2,\tau_*}}\bigr)^2.$$

	This finishes the proof of Theorem \ref{mth2}. q.e.d.

	\section{Decomposition of smooth functions} \label{decomp} In this section, we prove Lemma \ref{de1}. Thus assume that $Q$ is a spectral gap potential as in Section \ref{csg} and fix an arbitrary function $f\in C^\infty_b(\C)$.

	Let $G$ and $U$ be the bounded and the unbounded component of $\C\setminus S$ respectively, cf.~Figure \ref{fig2}. We must show that $f$ can be decomposed as
	$$f=\re g+f_0+\para\omega$$
	where
	$g$ is holomorphic and bounded in $G\cup U$, $f_0=0$ on $\d S$, and $\omega$ is a function with $\omega=0$ in a neighbourhood of $S_{\tau_*}$ and $\omega=1$ in a neighbourhood of $S\setminus S_{\tau_*}$; all functions are smooth.

	The case when $f$ is supported in a small neighbourhood of $U$ is is already treated in \cite[Lemma 5.1]{AM}, so we will here focus on the only remaining case of interest, i.e., when $f$ is supported in a small neighbourhood of the closure $\overline{G}$.

	Let $H(z)$ be the harmonic function on $G$ which solves the Dirichlet problem $H=f$ on $\d G$. We recall from e.g.~ \cite{B,GM} that $H(z)$ is smooth up to the boundary.
	
	As in \eqref{diric} we write $H=\re g_1+\para\varpi$ where
	$g_1(z)$ is holomorphic in $G$ and $\varpi(z)$ is the harmonic function on $G$ with boundary values $0$ on $C_1$ and $1$ on $C_2$; $\para$ is a real constant.

	Since $g_1(z)$ is smooth up to the boundary, we can extend this function to a smooth function $g(z)$ on $\C$, which is supported in a small neighbourhood of $\overline{G}$. (Whitney's extension theorem.)
	
	Likewise we extend the function $\varpi(z)$ to a smooth function $u(z)$
	which is supported in a small neighbourhood of $\overline{G}$.

	We finally extend $H(z)$ to $\C$ via $H:=\re g+\para u$ and define
	$$f_0:=f-H+\para u-\para\omega$$
	and observe that $f_0$ is smooth with $f_0=0$ on $\d S$. $\qed$

	\section{Gaussian fluctuations} \label{brief}

	In this section we prove Theorem \ref{m7}. Thus we assume that $Q$ is a spectral gap potential as in Section \ref{csg}.

	We assume that the test function $f\in C_b^\infty(\C)$ is of class $\calG$, i.e.,
	$$f=\re g+f_0=g/2+\bar{g}/2+f_0$$
	where $g$ is holomorphic in each connected component of $\hat{\C}\setminus S$ and $f_0=0$ on $\d S$; all functions are smooth on $\C$. We may assume that $f_0$ supported in a small neighbourhood of $S$.

	Following \cite{ACC1,AM}, we will write
	$$f_+:=g/2,\qquad  f_-:=\bar{g}/2.$$

	We shall next show that the proof in \cite{AM,ACC1} can be modified to the present situation. In the next section, we shall find that the more general Theorem \ref{main3} can be proven with only a slight extra twist.

	We begin by recalling (in a suitably adapted form) the main tool: the limit Ward identity.

	\subsection{Limit Ward identity} Let $v(z)$ be a suitable complex-valued function on $\C$.
	
	We assume throughout that $v$ is bounded, Lipschitz continuous and uniformly smooth in $\C\setminus (\d S)$. (The last condition means that both $v|_S$ and $v|_{\C\setminus S}$ are smooth up to the boundary.

	Next fix a smooth, bounded function $h$ and consider the potential
	$$\tilde{Q}=Q-\frac h n.$$

	Recall that $d\sigma=\Delta Q\cdot\1_S\, dA$ denotes the equilibrium measure.

	Now define functionals (signed measures) $\nu_n$ and $\tilde{\nu}_n$ by
	\begin{align*}\nu_n(f)=\E_n(\fluct_n f)=n(\sigma_n-\sigma)(f),\qquad
		\tilde{\nu}_n(f)=\tilde{\E}_n(\fluct_n f)=n(\tilde{\sigma}_n-\sigma)(f),
	\end{align*}
	where $\E_n$ is the expectation with respect to $Q$ and $\tilde{\E}_n$ is with respect to $\tilde{Q}$.

	For $z\in\C$ we write $k_z$ for the Cauchy kernel
	$$k_z(w)=\frac 1 {z-w}.$$

	Following \cite{AM} we now introduce two basic functions $D_n(z)$ and $\tilde{D}_n(z)$ by
	$$D_n(z)=\nu_n(k_z),\qquad \tilde{D}_n(z)=\tilde{\nu}_n(k_z).$$

	We have the following result, known as ``limit Ward identity''.

	\begin{lem} \label{skog}Let $v(z)$ be bounded, Lipschitz continuous and uniformly smooth in $\C\setminus (\d S)$. Then there is a number $\beta>0$ such that
		$$\int_\C[v\Delta Q+\dbar v\cdot \d(Q-\check{Q})]\tilde{D}_n\, dA=-\frac 1 2\sigma(\d v)-\sigma(v\cdot \d h)+\bigO(n^{-\beta}).$$
	\end{lem}

	A detailed proof is given in \cite[Proposition 1.2]{ACC1}. We remark that the proof given there is based on the previous analysis of connected droplets in \cite{AM}, using Ward's identity and some apriori estimates for the correlation kernel away from the boundary.

	\subsection{Asymptotics for $\nu_n(f)$}

	First consider the function $$p(z):=f_+(z)+f_0(z)=g(z)/2+f_0(z).$$

	Introduce functions $v_0,v_+,v_-$ by
	$$v_+=\frac {\dbar f_+}{\Delta Q},\qquad v_0=\frac {\dbar f_0}{\Delta Q}\1_S+\frac {f_0}{\d(Q-\check{Q})}\1_{\C\setminus S},\qquad v_-=\frac {\d f_-}{\Delta Q}=\bar{v}_+.$$

	These functions meet the requirements on $v$ in Lemma \ref{skog}.

	It is straightforward to check that for $v=v_0+v_+$,
	$$v\Delta Q+\dbar v\cdot \d(Q-\check{Q})=\dbar p\qquad \text{on}\qquad \C\setminus \d S.$$

	By Lemma \ref{skog}, we thus have, as $n\to\infty$,
	\begin{equation}\label{weve}\int_\C \dbar p\cdot \tilde{D}_n=-\frac 1 2\sigma(\d v)-\sigma(v\cdot \d h)+\bigO(n^{-\beta}).\end{equation}

	Hence, setting $h=0$,
	$$\nu_n(p)=\frac 1 2 \sigma(\d v)+\bigO(n^{-\beta}).$$
	Similarly (or by taking complex conjugates)
	$$\nu_n(f_-)=\frac 1 2 \sigma(\dbar v_-)+\bigO(n^{-\beta}).$$

	Now recall that $L(z)$ denotes a fixed smooth function which equals to $\log\Delta Q$ in a small neighbourhood of $S$ and vanishes identically outside of a slightly larger neighbourhood.

	By the definition of $v_0$ and straightforward computations,
	\begin{align*}\sigma(\d v_0)&=\int_S \Delta f_0-\int_S\dbar f_0\d\log \Delta Q\\
		&=\int_S\Delta f_0+\int_Sf_0\Delta L.\end{align*}

	Similarly, one verifies easily that
	$$\sigma(\d v_+)=\int_S\Delta f_+-\int_\C\dbar f_+\cdot\d L=\int_S\Delta f_++\int_\C f_+\Delta L$$
	and
	$$\sigma(\dbar v_-)=\int_S\Delta f_-+\int_\C f_-\Delta L.$$

	Adding up, and using $f^S=f_0\1_S+f_-+f_+$,
	\begin{align}\nu_n(f)&=\int_S \Delta f+\int_S f_0\Delta L+\int_\C (f_-+f_+)\Delta L+\bigO(n^{-\beta})\nonumber \\
		&=\int_S \Delta f+\int_\C f^S\Delta L\, dA+\bigO(n^{-\beta}).\label{u1}
	\end{align}

	\subsection{Asymptotics for $\tilde{\nu}_n(f)-\nu_n(f)$} Again decompose $f=f_++f_-+f_0$ and write $v=v_++v_0+v_-$. By Green's formula $\tilde{\nu}_n(f_+)=-\int \tilde{D}_n\cdot \dbar f_+$ we have
	$$-\tilde{\nu}_n(f_+)=\int_\C[v_+\Delta Q+\dbar v_+\cdot \d(Q-\check{Q})]\cdot \tilde{D}_n.$$
	Using the relationship \eqref{weve} for $p=f_0+f_+$ and $v=v_0+v_+$
	we obtain
	\begin{align*}\tilde{\nu}_n(p)-\nu_n(p)&=\sigma(v\cdot\d h)+\bigO(n^{-\beta})\\
		&=\int_S \dbar p\cdot\d h+\bigO(n^{-\beta}).
	\end{align*}

	By a similar argument (or just taking complex conjugates)
	$$\tilde{\nu}_n(f_-)-\nu_n(f_-)=\int_S\d f_-\cdot \dbar h+\bigO(n^{-\beta}).$$

	Adding the above, we obtain, as $n\to\infty$,
	\begin{equation}\label{u2}\tilde{\nu}_n(f)-\nu_n(f)=\int_S(\dbar f_+\cdot \d h+ \dbar f_0\cdot \d h+
		\d f_-\cdot \dbar h)+\bigO(n^{-\beta}).\end{equation}
	
	By a computation in \cite[p. 1191]{AM} we have
	\begin{equation}\label{u3}\int_S (\dbar f_+\cdot \d h+ \dbar f_0\cdot \d h+
		\d f_-\cdot \dbar h)=\frac 1 4 \int_\C \nabla (f^S)\bigcdot \nabla (h^S),\end{equation}
	the ``$\bigcdot$'' is the standard dot product in $\C=\R^2$.
	
	\subsection{Proof of Theorem \ref{m7}} \label{babbel} Fix a smooth function $f\in\calG$.

	We start by writing $F_{n,f}(t)=\log\E_n (e^{t\fluct_n f})$ in the form (cf.~\cite[Lemma 1.8]{ACC1})
	$$F_{n,f}(t)=\int_0^t\tilde{\E}_{n,sf}(\fluct_n f)\, ds$$
	where $\tilde{\E}_{n,h}$ is expectation with respect to $\tilde{Q}=Q-\frac 1 n h$.

	Now put
	$$e_f=\frac 1 2\int_S\Delta f+\frac 1 2\int_\C f^S\Delta L,\qquad v_f=\frac 1 4 \int_\C |\nabla (f^S)|^2.$$

	By \eqref{u2} and \eqref{u3}  with $h=sf$ and \eqref{u1} we have
	$$\tilde{\E}_{n,sf}(\fluct_n f)-e_f=sv_f+\bigO(n^{-\beta}).$$

	Integrating, we obtain
	$$F_{n,f}(t)=te_f+\frac {t^2} 2 v_f+\bigO(n^{-\beta}).$$

	It remains to prove that $e_f$ is given in terms of Neumann's jump as in \eqref{ef}. We refer to \cite[Section 6.5]{ACC1} for a detailed derivation of this relationship. q.e.d.

	\section{Proof of Theorem \ref{main3}} \label{Pmain3}

	Take $f\in C^\infty_b(\C)$ and write $f=f_1+f_2$ where
	$f_1\in\calG$ and $f_2=\para\omega\in\calH$, according to Lemma \ref{de1}.
	
	Consider the joint cgf
	$$F_n(s,t):=\log \E_n \exp(s\fluct_n f_1+t\fluct_n f_2).$$
	
	Also write
	$$h:=sf_1+tf_2.$$
	
	By Theorem \ref{mth2} we have that
	\begin{equation}\label{chob}
		F_n(0,t)=F_{\para X_n^+}(t)+F_{-\para X_n^-}(t)+\bigO(\delta_n),
	\end{equation}
	where $X_n^+$, $X_n^-$ have the required Heine distributions and the implied constant is uniform for $|t|\le \log n$.
	
	Also, in notation of Section \ref{babbel},
	$$\frac {\d F_n}{\d s}=\tilde{\E}_{n,h}(\fluct_n f_1),$$
	so
	$$F_n(t,t)=F_n(0,t)+\int_0^t \tilde{\E}_{n,h}(\fluct_n f_1)\, ds.$$
	
	But by \eqref{u1} and \eqref{u2}, since $f_2$ is locally constant on $S$
	\begin{align}\label{box}\tilde{\E}_{n,h}(\fluct_n f_1)=e_{f_1}+sv_{f_1}
		+ \bigO(n^{-\beta}).\end{align}

	Integrating in \eqref{box} and using \eqref{chob} we obtain
	\begin{align*}F_n(t,t)&=F_{\para X_n^+}(t)+F_{-\para X_n^-}(t)+F_Y(t)+\bigO(n^{-\beta}),\end{align*}
	where $Y$ is Gaussian with expectation $e_{f_1}$ and variance $v_{f_1}$. q.e.d.

\subsection*{Acknowledgements} The authors thank the Institut Mittag-Leffler for the hospitality during the program "Random Matrices and Scaling Limits" where this work was partly conducted. The second author acknowledges support from the Royal Swedish Academy of Sciences G.S. Magnuson's fund.

\end{document}